\documentclass[lettersize,journal,twoside]{IEEEtran}
\usepackage{amsmath,amssymb,amsfonts}
\usepackage{array}
\allowdisplaybreaks[4]
\usepackage[caption=false,font=normalsize,labelfont=sf,textfont=sf]{subfig}
\usepackage{textcomp}
\usepackage{stfloats}
\usepackage{url}
\usepackage{verbatim}
\usepackage{graphicx}
\usepackage{cite}
\usepackage{makecell}
\usepackage{graphicx,color}
\usepackage{subfig}
\usepackage{multirow}
\usepackage{textcomp}
\usepackage{booktabs}
\usepackage{balance}
\usepackage[linesnumbered, ruled]{algorithm2e}
\SetKwRepeat{Do}{do}{while}%
\usepackage[thmmarks,amsmath]{ntheorem}
\theoremseparator{:}

\newtheorem{lemma}{Lemma}
\newtheorem{proposition}{Proposition}
\theoremheaderfont{\it}
\theorembodyfont{\normalfont}
\theoremsymbol{\ensuremath{\Box}}
\newtheorem*{proof}{Proof}
\theoremheaderfont{\it}
\theorembodyfont{\normalfont}
\theoremsymbol{}
\newtheorem{remark}{Remark}
\usepackage{setspace}
\begin{document}

\title{Physical Layer Security for Sensing-Communication-Computing-Control Closed Loop: A Systematic Security Perspective}

\author{Chengleyang Lei, Wei Feng,~\IEEEmembership{Senior Member,~IEEE}, Yunfei Chen,~\IEEEmembership{Fellow,~IEEE}, Jue Wang,~\IEEEmembership{Member,~IEEE}, \\Ning Ge,~\IEEEmembership{Member,~IEEE}, Shi Jin,~\IEEEmembership{Fellow,~IEEE}, and Tony Q. S. Quek,~\IEEEmembership{Fellow,~IEEE}
	\thanks{The work of W. Feng was supported in part by the National Natural Science Foundation of China under Grant 62541104, Grant 62425110, and Grant U22A2002, in part by the Key Research and Development Project of Nantong (Special Project for Prospective Technology Innovation) under Grant GZ2024002, and in part by the FAW Jiefang Automotive Co., Ltd. The work of J. Wang was supported in part by the National Natural Science Foundation of China under Grant 62171240. The work of T. Q. S. Quek was supported in part by the National Research Foundation, Singapore and Infocomm Media Development Authority under its Communications and Connectivity Bridging Funding Initiative. Any opinions, findings and conclusions or recommendations expressed in this material are those of the authors and do not reflect the views of National Research Foundation, Singapore.}
	\thanks{Chengleyang Lei, Wei Feng (corresponding author), and Ning Ge are with the Department of Electronic Engineering, State Key Laboratory of Space Network and Communications, Tsinghua University, Beijing 100084, China (email: lcly21@mails.tsinghua.edu.cn, fengwei@tsinghua.edu.cn, gening@tsinghua.edu.cn).}
	\thanks{Yunfei Chen is with the Department of Engineering, University of Durham, DH1 3LE Durham, U.K. (e-mail: yunfei.chen@durham.ac.uk).}
	\thanks{Jue Wang is with the School of Information Science and Technology, Nantong University, Nantong 226019, China (email: wangjue@ntu.edu.cn).}
	\thanks{Shi Jin is with the National Mobile Communications Research Laboratory, Southeast University, Nanjing 210096, China (email: jinshi@seu.edu.cn).}
	\thanks{T. Q. S. Quek is with the Singapore University of Technology and Design, Singapore 487372, and also with the Department of Electronic Engineering, Kyung Hee University, Yongin 17104, South Korea (e-mail: tonyquek@sutd.edu.sg).}
}
        % <-this % stops a space
%\thanks{}% <-this % stops a space

%\markboth{IEEE WIRELESS COMMUNICATIONS LETTERS}
%{Control-Oriented Power Allocation for Integrated Satellite-UAV Networks}



\maketitle

\begin{abstract}
In industrial automation or emergency rescue, sensors and robots work together with the help of an edge information hub (EIH) containing both communication and computing modules. Typically, the EIH collects the sensing data via the sensor-to-EIH link, processes data and then makes decisions on board before sending commands to the robot via the EIH-to-robot link. This forms a sensing-communication-computing-control ($\textbf{SC}^3$) closed loop. In practice, the inherent openness of wireless links within the closed loop leads to susceptibility to eavesdropping. To this end, this paper refines the conventional physical layer security (PLS) approach with a systematic thinking to safeguard the $\textbf{SC}^3$ closed loop. The closed-loop negentropy (CNE), a new metric for the performance of the whole $\textbf{SC}^3$ closed loop, is maximized under the closed-loop security constraint. The transmit time, power, bandwidth of both wireless links, and the computing capability, are jointly designed. The optimization problem is non-convex. We leverage the Karush-Kuhn-Tucker (KKT) conditions and the monotonic optimization (MO) theory to derive its globally optimal solution. Simulation results show the performance gain of the proposed systematic approach, and reveal the advantage of exploiting the closed-loop structure-level PLS over the link-level or sum-link-level designs.
\end{abstract}


\section{Introduction}
In hazardous or inaccessible environments, autonomous machines and robots have been increasingly utilized to execute mission-critical tasks that are dangerous or dull for humans. In applications such as emergency rescue~\cite{rescue_robot}, scientific exploration~\cite{scientific_robot}, and remote industrial automation~\cite{industrial_robot}, the deployment of robots can help to reduce risks to humans and improve efficiency. Such tasks rely on communications networks to connect the sensors and robots, facilitating the exchange of vital information like sensor data and control commands. However, in remote or post-disaster areas, terrestrial networks are often unavailable due to the harsh geographical conditions. Non-terrestrial networks (NTN), which consist of satellites and autonomous aerial vehicles (AAVs, also known as UAVs), are vital infrastructure to support these remote tasks due to the advantages of wide coverage and flexible deployment~\cite{network,network1,network2}. In such cases, UAVs can be equipped with devices integrating communication and computing modules for information collection and processing, referred to as edge information hubs (EIHs). These EIHs provide edge communications and computing support for ground robots, while utilizing a dedicated satellite link as backhaul to connect to a remote control center for controlling or computing offloading\cite{jsac}. 

However, the integration of NTNs introduces new security risks. Due to the openness and long distance of air-ground links, sensitive information, including sensor data and control commands, is vulnerable to eavesdropping~\cite{satellite_security, UAV_security}. While traditional cryptographic methods offer robust data protection, their high computational complexity may introduce additional latency and energy consumption~\cite{PLS_survey1}, which may be detrimental in time-sensitive and resource-constrained closed-loop control scenarios. Federated learning has also emerged as a privacy-preserving approach at the computing layer, allowing distributed devices to learn a shared model without uploading raw data~\cite{FL1, FL2}. However, the information regarding control commands and sensing features will still be transmitted via radio channels, necessitating physical-layer security measures~\cite{FL3}. Recently, physical-layer encryption (PLE) techniques, such as constellation obfuscation, have attracted increasing interest~\cite{PLE1,PLE2}. By obfuscating transmitted symbols on the physical layer, PLE remains effective even with a high eavesdropper signal-to-noise ratio~\cite{PLE1}. However, typical PLE methods have high key-to-data ratios, which may introduce high computational complexity for the hardware design~\cite{PLE3}. As an alternative, physical layer security (PLS) techniques are recognized as a promising approach for safeguarding NTN-based communications. By exploiting the characteristics of wireless channels and designing transmission parameters, PLS techniques can improve the channel advantage of legitimate links over eavesdropping links, thereby directly enhancing the communication security~\cite{PLS_survey2, PLS_survey3,PLS_survey4}.   

Unlike traditional point-to-point communication systems, NTN-based control systems typically incorporate multiple coupled communication links within an integrated  sensing-communication-computing-control ($\textbf{SC}^3$) closed loop~\cite{network}. This brings unique challenges when applying PLS techniques. Specifically, an $\textbf{SC}^3$ closed loop involves two critical communication links: the sensor-to-EIH link, which sends sensor data to the EIH for analysis, and the EIH-to-robot link, which transmits the control commands from the EIH to the robot to guide its action. Information leakage on either link could be harmful to the overall security. Therefore, the two links should be considered as a unified entity, necessitating a systematic security assessment. However, current PLS techniques primarily focus on link-level or sum-link-level security metrics (e.g., the secrecy rate of each individual link), without a comprehensive thinking of coupled links within an $\textbf{SC}^3$ closed loop. Motivated by the above observations, in this paper, we consider the security performance of the whole $\textbf{SC}^3$ closed loop from a systematic perspective. Specifically, we establish a holistic closed-loop security constraint on the total leaked task-relevant information, under which we jointly design the communication and computing resource allocation, including bandwidth, power, time, and computing capability, to improve the overall closed-loop control performance.

\subsection{Related Works}      
Networked control systems (NCSs), where the components of control systems are connected through communication networks, have been applied widely due to the advantages of flexibility and maintainability~\cite{NCS_survey1, NCS_survey2,NCS1}. The relationship between the communication and control components in NCSs has been investigated in many works~\cite{NCS2,NCS3,LQR}. In \cite{NCS2},  a lower bound on the communication data rate to stabilize a linear control system was provided, which was related to the state matrix of the system. In \cite{NCS3}, the authors proposed the optimal control policy to minimize the directed information transmitted, considering a constraint on control performance measured by the linear quadratic regulator (LQR) control cost. The authors in \cite{LQR} further derived a lower bound on the communication data rate to achieve a desired LQR cost. Based on the relationship between the communication and control components, some research has been conducted on the system design aiming at improving the overall performance of $\textbf{SC}^3$ closed loops~\cite{NCS4,single_loop}. The authors in \cite{NCS4} investigated a wireless NCS supported by a space-air-ground integrated network. The transmission policy was optimized to minimize the weighted sum of control cost and power consumption. In \cite{single_loop}, the authors proposed a joint uplink and downlink design framework, where the LQR cost was minimized by jointly optimizing the uplink and downlink transmit power, time, bandwidth, and the CPU frequency within an $\textbf{SC}^3$ closed loop. The authors proposed a new metric named closed-loop negentropy (CNE) to measure the task-relevant information transmitted through an $\textbf{SC}^3$ closed loop. These studies provided great insight into the design of the $\textbf{SC}^3$ closed loop. However, most of them did not consider the security issues.

Some research has been conducted on the security of NCSs~\cite{NCS_security1,NCS_security2,NCS_security3}. The authors in \cite{NCS_security1} proposed an energy-efficient security-aware architecture for NCSs under deception attacks. The selective encryption method was adopted to save energy and detect the attack. In \cite{NCS_security2}, the resilient control in a wireless NCS was considered under a denial-of-service attack. The authors derived the Nash power strategies and optimal control strategy, and then proposed a pricing mechanism design approach to keep the system stable. Yang \textit{et al.} in \cite{NCS_security3} exploited the variable sampling approach to deal with the dropouts caused by jamming attacks and designed a resilient variable sampling controller by a delta operator method. The above works mainly considered active attacks on NCSs. The authors in \cite{NCS_security5} investigated an NCS of permanent magnetic linear motor under false data injection attacks. They proposed an observer-based non-singular terminal sliding-mode controller, where an extended state observer was used to estimate the system states as well as the attack/disturbance terms. In \cite{NCS_security6} a predictor-based decentralized periodic event-triggered output-feedback secure control framework for nonlinear large-scale NCSs was proposed to guarantee good control performance under replay attacks while reducing communication bandwidth. The countermeasures to passive eavesdropping are also important to protect the confidentiality of critical tasks. In \cite{NCS_security4}, the authors proposed a security provider to encrypt the data in power line communications networks for metering and automation control systems. In \cite{NCS_security7}, a bilinear-map-based key encapsulation mechanism combined with symmetric encryption was proposed to protect a distributed model predictive control system against passive eavesdropping. However, the cryptography-based approaches introduce additional complexity and latency, which are unsuitable for time-sensitive closed-loop control applications. 

The PLS techniques have been widely utilized to safeguard NTN-based communications~\cite{Zhang2019,PLS1,PLS2,PLS3}. Zhang \textit{et al.} in \cite{Zhang2019} considered both the UAV-to-ground and ground-to-UAV communications subject to a ground-based eavesdropper. The UAV's trajectory and the transmit power of the transmitter were jointly optimized to maximize the average secrecy rates. Reference \cite{PLS1} introduced a cooperative UAV that transmits jamming signals for safeguarding the UAV-to-ground communications. The minimum secrecy rate was maximized by jointly optimizing the UAV trajectory, the transmit power, and user scheduling. The authors in \cite{PLS2} further investigated three-dimensional trajectory design to maximize the worst-case secrecy rate, considering imperfect eavesdropper location information. In \cite{PLS3}, a satellite-supported Internet of Things Network with a UAV relay was investigated. A max–min secrecy rate optimization problem was formulated where the power allocation, UAV beamforming, and the UAV position were jointly optimized. Most of these works focus on single-link security metrics, such as the secrecy rate. However, in the $\textbf{SC}^3$ closed loop, the sensor-to-EIH link and EIH-to-robot link are inherently coupled, with their transmitted information being closely related. Therefore, to effectively secure the $\textbf{SC}^3$ closed loop, a more systematic view that considers the coupled links within the loop as a single entity should be adopted.
      
\subsection{Main Contributions and Organization}
Motivated by the above considerations, in this paper, we investigate an NTN-based $\textbf{SC}^3$ closed loop in the presence of a ground-based eavesdropper attempting to intercept transmissions of both sensor data and control commands. We maximize the CNE of the $\textbf{SC}^3$ closed loop by jointly optimizing the communication and computing parameters of the closed loop, with a constraint on the total task-relevant information leakage. We propose efficient algorithms to obtain the globally optimal solution to the optimization problem and provide simulation results to validate the proposed algorithms. The main contributions are summarized as follows.
\begin{itemize}
	\item We investigate an NTN-assisted $\textbf{SC}^3$ closed loop, consisting of a sensor, a UAV-mounted EIH, and a robot, in the presence of a potential eavesdropper attempting to intercept the confidential task information. Considering the holistic security of the $\textbf{SC}^3$ closed loop, we establish a closed-loop security constraint such that the total task-relevant information leakage should not exceed a threshold. Under this constraint, we jointly optimize the bandwidth, power, and transmission times for both the sensor-to-EIH link and EIH-to-robot link, as well as the computing capability, to maximize the CNE of the $\textbf{SC}^3$ closed loop.
	\item The formulated CNE maximization problem is non-convex due to the coupling among the time, power, and bandwidth. We first simplify the original problem by deriving the optimal configurations for computing capability and transmission time. Subsequently, the recast problem is analyzed under distinct channel conditions. For the superior legitimate channel case, we analyze properties of the globally optimal solution based on Karush-Kuhn-Tucker (KKT) conditions and propose an effective algorithm to find the solution. For other channel conditions, we transform the problem to a monotonic optimization (MO) problem, and utilize the MO theory to find its globally optimal solution.
	\item We provide simulation results to show the performance gain of the proposed scheme over the traditional link-oriented system design method, revealing the superiority of the closed-loop security framework. In addition, we discuss the influence of key system parameters on the achievable CNE via simulations, offering insights into the secure $\textbf{SC}^3$ closed loop design.
\end{itemize}

We further clarify the difference between this work and our prior work~\cite{single_loop}. While sharing the foundational $\textbf{SC}^3$ closed loop model, this work incorporates the closed-loop security considerations and thus introduces a fundamentally new performance–security tradeoff. In \cite{single_loop}, the optimal solution fully utilizes transmit power, and the design reduces to a convex bandwidth allocation to achieve an uplink–downlink balance. In this work, simply increasing transmit power will also increase the eavesdropping risk. Therefore, the power and bandwidth need to be designed carefully. Moreover, the non-convex coupling among the time, bandwidth, and power makes the original method in \cite{single_loop} no longer applicable, which motivates a new globally optimal solution based on KKT analysis and MO. In addition, we reveal a new theoretical insight: for each link, the optimal solution depends on the relative ordering of the legitimate and eavesdropping channel gains, which will be shown in \textbf{Proposition \ref{Prop2}}.

The rest of this paper is organized as follows. In Section \ref{Section:System_Model}, we introduce the model of the considered NTN-based control system, and formulate a CNE maximization problem considering the closed-loop security constraint. In Section \ref{Section:Simplification}, the optimization problem is recast to a more tractable form. In Section \ref{Section:Algorithm}, we propose algorithms to obtain its globally optimal solution. Simulation results are provided in Section \ref{Section:Simulation} with further discussions. Finally, we conclude this paper in Section \ref{Section:Conclusion}. For ease of reference, the abbreviations used in this paper are listed in Table~\ref{tab1}.

\begin{table}[!t]
	\centering
	\caption{Summary of main abbreviations.}
	\label{tab1}
	\begin{tabular}{cp{5.3cm}}\toprule
		\textbf{Abbreviation} & \textbf{Full name}\\\midrule
		AAV & autonomous aerial vehicle \\\hline
		CNE & closed-loop negentropy \\\hline
		EIH & edge information hub \\\hline
		KKT & Karush-Kuhn-Tucker \\\hline
		LICQ & linear independence constraint qualification \\\hline
		LoS & line-of-sight \\\hline
		LQR & linear quadratic regulator \\\hline
		MO & monotonic optimization  \\\hline
		NCS & networked control system \\\hline
		NTN & non-terrestrial network \\\hline
		PLE & physical layer encryption \\\hline
		PLS & physical layer security \\\hline
		$\textbf{SC}^3$ & sensing-communication-computing-control 
		\\\bottomrule
	\end{tabular}
\end{table}

\section{System Model and Problem Formulation}
\label{Section:System_Model}
\begin{figure} [t]
	\centering
	\includegraphics[width=0.9\linewidth]{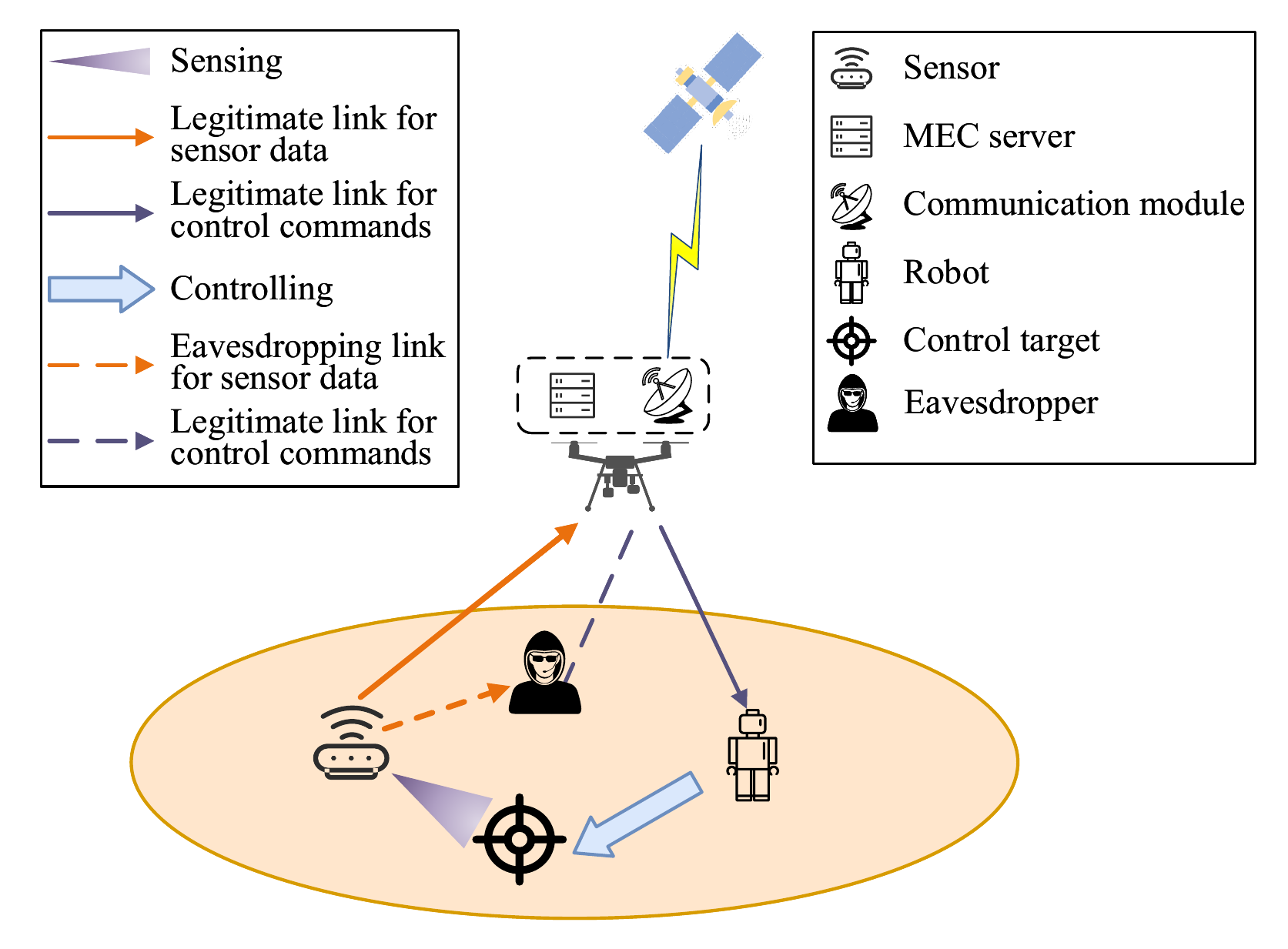}
	\caption{Illustration of an $\textbf{SC}^3$ closed loop consisting of a sensor, an EIH, and a robot. The EIH is elevated on a UAV, and telecontrolled by a dedicated satellite link. One eavesdropper wiretaps both the sensor-to-EIH and EIH-to-robot links, to accordingly crack the task information.}
	\label{fig:system}
\end{figure}
\subsection{System Model}
As shown in Fig. \ref{fig:system}, we consider an NTN-assisted control system. A UAV-mounted EIH integrating communication and computing modules provides edge computation and communication services for a ground robot executing a control task. The EIH utilizes a dedicated satellite link to connect to a remote control center for real-time monitoring and telecontrol. The sensor, the EIH, and the robot constitute an $\textbf{SC}^3$ closed loop for task completion. Specifically, within the loop, the sensor acquires information regarding the control target and transmits it to the EIH. Subsequently, the EIH processes the sensor data to generate proper control commands. Next, the control commands are transmitted to the robot to guide its actions. The whole process operates periodically within the $\textbf{SC}^3$ closed loop. Furthermore, the security implications of a potential ground-based eavesdropper are considered. The position of the eavesdropper can be detected by a UAV-mounted camera or radar~\cite{position1,position2}, which is assumed to be known at the EIH.\footnote{In practice, the eavesdropper’s position may only be known approximately. In such a case, a common approach is to model the eavesdropper within a bounded uncertainty region (e.g., a circular area) around the estimated location~\cite{position3}. With the position uncertainty, the proposed framework can be extended to a robust design that considers the worst-case secrecy performance, which is an interesting direction for future work.} The eavesdropper attempts to intercept transmissions on both the sensor-to-EIH link and the EIH-to-robot link to obtain task-relevant information. In the following, we will present detailed models for the $\textbf{SC}^3$ closed loop and the eavesdropping process. 

The sensor and robot can be regarded as a single virtual user, as they cooperate to accomplish a common control task. Consequently, the sensor-to-EIH link and the EIH-to-robot link can be regarded as the uplink and downlink of this virtual user. Denoting the transmit power, time, and bandwidth of the two links as $p_{\text{u/d}},t_{\text{u/d}},B_{\text{u/d}}$, respectively, where the subscripts $\text{u}$ and $\text{d}$ represent the uplink and downlink, respectively, the corresponding data rates can be calculated as
\begin{equation}
	R_{\text{u/d}}\left( p_{\text{u/d}},B_{\text{u/d}}\right)  = B_{\text{u/d}}\log_2\left( 1+\frac{p_{\text{u/d}}g_{\text{u/d}}}{B_{\text{u/d}}N_0}\right),
\end{equation}
where $g_{\text{u}}$ and $g_{\text{d}}$ denote the large-scale channel gains of the uplink and downlink, and $N_0$ denotes the noise power spectral density. We assume that air-ground channels are dominated by the line-of-sight (LoS) component and thus the effects of small-scale fading are negligible \cite{Zhang2019}. Due to the limited system resources, we have the following constraints
\begin{align}
	&t_{\text{u}}+t_{\text{c}}+t_{\text{d}}\leq T,\label{time_constrain}\\
	&p_{\text{u}} \leq P_{\text{umax}},\quad p_{\text{d}} \leq P_{\text{dmax}},\\
	&B_{\text{u}} \leq B_{\text{max}},\quad B_{\text{d}} \leq B_{\text{max}},
\end{align}
where $t_{\text{c}}$ denotes the computation time, $T$ represents the total time of each control cycle, $P_{\text{umax}}$ denotes the uplink transmit power constraint of the sensor, $P_{\text{dmax}}$ denotes the downlink transmit power constraint of the EIH, and $B_{\text{max}}$ denotes the bandwidth constraint. Notably, as implied by \eqref{time_constrain}, uplink and downlink transmissions occupy distinct time intervals, allowing them to share the same frequency.

Denoting the uplink and downlink data throughput per control cycle as $D_{\text{u}}$ and $D_{\text{d}}$, we have
\begin{equation}
	D_{\text{u/d}} \leq t_{\text{u/d}}R_{\text{u/d}}\left( p_{\text{u/d}},B_{\text{u/d}}\right).
\end{equation}

The EIH analyzes the received sensor data, extracts task-relevant information, then generates control commands. The computation time can be calculated as $t_{\text{c}}=\frac{\alpha D_{\text{u}}}{f}$, where $\alpha$ denotes the required CPU cycles to process one-bit sensor data, and $f$ is the computing capability allocated for processing the sensor data, constrained by the CPU frequency $f_{\text{max}}$ of the EIH, i.e., $f \leq f_{\text{max}}$.

The computing process can be regarded as an information-extraction process from the sensor data. Following the model in \cite{single_loop}, the extracted task-relevant information is given by $\rho D_{\text{u}}$, where $\rho\in\left( 0,1\right) $ denotes the information extraction ratio. The parameter $\rho$ characterizes the efficiency of the EIH in extracting task-relevant features from raw sensing data. Its value depends on three key factors: the sensor modality, the task type, and the adopted feature-extraction algorithm. For example, high-dimensional data sources, such as video or LiDAR point clouds, typically contain significant redundancy, corresponding to a smaller $\rho$. In our setting, these factors are fixed; hence, $\rho$ can be treated as a constant. In practical applications, it can be obtained empirically by measuring the average size of the extracted features relative to the size of the raw input. The received amount of task-relevant information on the robot per control cycle, denoted by $D_{\textbf{SC}^3}$, is constrained by both the extracted information and the downlink data throughput, expressed as
\begin{equation}
	D_{\textbf{SC}^3}\leq \min \left( \rho D_{\text{u}},D_{\text{d}}\right).
\end{equation}
The metric $D_{\textbf{SC}^3}$ is referred to as the CNE, reflecting its role in reducing the overall entropy or uncertainty of the control system state through the $\textbf{SC}^3$ loop per cycle. The CNE is directly correlated with the closed-loop control performance of the system. Specifically, according to \cite{single_loop}, the lower bound of the LQR cost is related to the CNE as follows:
	\begin{align}\label{LQR_D}
		l \geq \frac{n  N \!\left( \mathbf{v}\right)|\det \mathbf{M}|^\frac{1}{n}} {2^{\frac{2}{n}\left( 	D_{\textbf{SC}^3} - \log_2 |\det \mathbf{A}|\right) }-1}+\text{tr}\left( \mathbf{\Sigma}_{\text{V}}\mathbf{S}\right),
	\end{align}
	where $l$ denotes the LQR cost of a control system, $n$, $\mathbf{v}$, $\mathbf{M}$, $\mathbf{A}$, $\mathbf{\Sigma}_{\text{V}}$, and $\mathbf{S}$ are control-related parameters, as defined in \cite{single_loop}. The LQR cost is a commonly used metric in optimal control theory~\cite{LQR}, which comprehensively measures the control state convergence and energy consumption. A smaller LQR cost indicates a better control performance. It can be observed from \eqref{LQR_D} that the lower bound of the LQR cost decreases monotonically with the CNE $D_{\textbf{SC}^3}$. Therefore, within this framework, improving the control performance is equivalent to maximizing the CNE.

The eavesdropper intercepts task-relevant information from both the sensor-to-EIH link and the EIH-to-robot link. As the sensor and the eavesdropper are on the ground, the sensor-to-eavesdropper channel consists of both large-scale path loss and small-scale Rayleigh fading. The channel coefficient is modeled as $h_{\text{SE}} = \xi \sqrt{g_{\text{SE}}}$, where $\xi \sim \mathcal{CN}(0,1)$ represents the small-scale channel fading, and $g_{\text{SE}}$ denotes the large-scale channel gain. The expected data rate from the sensor to the eavesdropper can be calculated as~\cite{Zhang2019}
\begin{align}
	R_{\text{SE}}\left( p_{\text{u}},B_{\text{u}}\right)  = &\mathbb{E}_{\xi}\left[ B_{\text{u}}\log_2\left( 1+\frac{p_{\text{u}}|\xi|^2g_{\text{SE}}}{B_{\text{u}}N_0}\right)\right]\label{RSEa}\\
	\leq &B_{\text{u}}\log_2\left( 1+\frac{p_{\text{u}}g_{\text{SE}}}{B_{\text{u}}N_0}\mathbb{E}_{\xi}\left[|\xi|^2 \right]\right) \label{RSEb}\\
	= &B_{\text{u}}\log_2\left( 1+\frac{p_{\text{u}}g_{\text{SE}}}{B_{\text{u}}N_0}\right),\label{RSEc}
\end{align}
where $\mathbb{E}_{\xi}$ denotes the expectation operation over $\xi$, and \eqref{RSEb} is obtained by applying Jensen’s inequality. Equation \eqref{RSEc} provides an upper bound on the expected data rate $R_{\text{SE}}$. For a conservative security analysis, this paper considers the worst case and assumes that the eavesdropper can achieve this upper bound, i.e., $R_{\text{SE}}\left( p_{\text{u}},B_{\text{u}}\right)  =  B_{\text{u}}\log_2\left( 1+\frac{p_{\text{u}}g_{\text{SE}}}{B_{\text{u}}N_0}\right)$. The impact of this approximation gap caused by small-scale fading will be further evaluated in Section \ref{Section:Simulation}. The total amount of sensor data information eavesdropped from the sensor-to-EIH link per control cycle can be calculated as $D_{\text{SE}} = t_{\text{u}}R_{\text{SE}}\left( p_{\text{u}},B_{\text{u}}\right)$.

Due to the air-to-ground transmission characteristics, the EIH-to-eavesdropper link is assumed to be dominated by a LoS path~\cite{Zhang2019}. The total amount of information about control commands eavesdropped from the EIH-to-robot link, denoted as $D_{\text{CE}}$, can be calculated as
\begin{equation}
	D_{\text{CE}} = t_{\text{d}}R_{\text{CE}}\left( p_{\text{d}},B_{\text{d}}\right),
\end{equation}
where $R_{\text{CE}}\left( p_{\text{d}},B_{\text{d}}\right)$ denotes the data rate from the EIH to the eavesdropper, formulated as
\begin{equation}
	R_{\text{CE}}\left( p_{\text{d}},B_{\text{d}}\right)  =  B_{\text{d}}\log_2\left( 1+\frac{p_{\text{d}}g_{\text{CE}}}{B_{\text{d}}N_0}\right),
\end{equation}
where $g_{\text{CE}}$ denotes the channel power gain of the eavesdropping link on control commands.

To ensure system security, we impose a constraint on the total task-relevant information leaked to the eavesdropper via both the sensor-to-EIH link and the EIH-to-robot link. Considering that the sensor data contains raw measurements and only a fraction $\rho$ of the raw data is task-relevant, the closed-loop security constraint can be formulated as
\begin{equation}
	\rho D_{\text{SE}} + D_{\text{CE}} \leq D_{\text{th}},
\end{equation}
where $D_{\text{th}}$ is the closed-loop information leakage threshold.

\subsection{Problem Formulation}
In this paper, we maximize the CNE under the closed-loop security constraint by jointly optimizing the bandwidth, power, and transmission times of the uplink and downlink, as well as the computing capability of the EIH. The optimization problem is formulated as
	\begin{subequations}\label{P1}
	\begin{align} \max\limits_{\substack{p_{\text{u}},t_{\text{u}},B_{\text{u}},f,\\p_{\text{d}},t_{\text{d}},B_{\text{d}}}} &D_{\textbf{SC}^3} \label{P1a}\\
		\text{s.t.} \ &D_{\textbf{SC}^3}\leq \min \left( \rho D_{\text{u}},D_{\text{d}}\right) , \label{P1b}\\
		&D_{\text{u/d}}\leq t_{\text{u/d}}R_{\text{u/d}}\left( p_{\text{u/d}},B_{\text{u/d}}\right) , \label{P1c}\\
		&t_{\text{u}}+\frac{\alpha D_{\text{u}}}{f}+t_{\text{d}}\leq T, \label{P1d}\\
		&0 \leq B_{\text{u}}\leq B_{\max},0 \leq B_{\text{d}}\leq B_{\max},  \label{P1e}\\
		&0 \leq p_{\text{u}}\leq P_{\text{umax}}, 0 \leq p_{\text{d}}\leq P_{\text{dmax}}, \label{P1f}\\
		&0 \leq f\leq f_{\max}, \label{P1g}
		\\&\rho D_{\text{SE}} + D_{\text{CE}} \leq D_{\text{th}},\label{P1h}\\
		&D_{\text{SE}} = t_{\text{u}}R_{\text{SE}}\left( p_{\text{u}},B_{\text{u}}\right),\label{P1i}\\
		&D_{\text{CE}} = t_{\text{d}}R_{\text{CE}}\left( p_{\text{d}},B_{\text{d}}\right),\label{P1j}
	\end{align}
\end{subequations}
where \eqref{P1d}-\eqref{P1g} denote the constraints for the time, bandwidth, power, and computing capability resources, and \eqref{P1h} represents the closed-loop security constraint. The optimization problem \eqref{P1} is non-convex due to the coupling between the time, power and bandwidth in \eqref{P1c}, \eqref{P1i}, \eqref{P1j}. In the next section, we simplify \eqref{P1} to a more tractable form.
 
\section{Problem Simplification}
\label{Section:Simplification}
First, we analyze whether the closed-loop security constraint \eqref{P1h} is active at the optimum. If constraint \eqref{P1h} is omitted, \eqref{P1} reduces to a general CNE maximization problem without security considerations. Denoting the optimal data throughputs for the uplink, downlink, sensor data eavesdropping link, and control command eavesdropping link as $D_{\text{u},0},D_{\text{d},0}, D_{\text{SE},0}$ and $D_{\text{CE},0}$, respectively, by solving problem \eqref{P1} without constraint \eqref{P1h}. The corresponding CNE achieved is $D_{\textbf{SC}^3,0} = \min \left( \rho D_{\text{u},0},D_{\text{d},0}\right)$, and the task-relevant information leaked to the eavesdropper is calculated as $D_{\text{E},0} = \rho D_{\text{SE},0} + D_{\text{CE},0}$. If the predefined security threshold is greater than or equal to the information leaked in the unconstrained optimal solution, i.e., $D_{\text{th}}\geq D_{\text{E},0}$, the closed-loop security constraint is not needed. Correspondingly, problem \eqref{P1} becomes equivalent to the aforementioned reduced problem (i.e., the CNE maximization problem without \eqref{P1h}), and the optimal CNE is $D_{\textbf{SC}^3,0}$. If $D_{\text{th}}< D_{\text{E},0}$, the closed-loop security constraint should be active and must be satisfied with equality, i.e., 
\begin{align}\label{security_equal}
\rho D_{\text{SE}} + D_{\text{CE}} = D_{\text{th}},\quad \text{if} \  D_{\text{th}}< D_{\text{E},0}.
\end{align}
The optimal solution to the reduced problem without \eqref{P1h} has been previously given in \cite{single_loop}. Therefore, in the following, the subsequent analysis will focus on the more intricate case where $D_{\text{th}}< D_{\text{E},0}$, implying that the closed-loop security constraint \eqref{security_equal} is active. 
\begin{lemma}\label{Lemma1}
	The optimal solutions to \eqref{P1} satisfy the equations
	\begin{align}
		&f^* = f_{\text{max}},\label{lemma1_1}\\
		&\rho D^*_{\text{u}}=D^*_{\text{d}}\label{lemma1_2}.
	\end{align}
\end{lemma}
\begin{proof}
	We prove this lemma by contradiction. If $f < f_{\text{max}}$, improving the computing capability to $f = f_{\text{max}}$ will relax the constraint \eqref{P1d} and keep the objective function and other constraints unchanged. Therefore, we can assume that $f^* = f_{\text{max}}$. Similarly, if the communication capability of the uplink and downlink is not balanced, e.g., assuming that $\rho D_{\text{u}}<D_{\text{d}}$, we can reduce the downlink time $t_\text{d}$ until $\rho D_{\text{u}}=D_{\text{d}}$, so as to relax the constraints \eqref{P1d} and \eqref{P1j}, while keeping the objective function and other constraints unchanged. If $\rho D_{\text{u}}>D_{\text{d}}$, we can relax the constraints in a similar way. Therefore, we have $\rho D^*_{\text{u}}=D^*_{\text{d}}$, which completes the proof.
\end{proof}

Based on \textbf{Lemma \ref{Lemma1}}, we derive the optimal transmission time allocation and an equivalent transformation of the original optimization problem, shown in the following proposition.

\begin{proposition}\label{Prop1}
	For the optimization problem \eqref{P1}, assuming an active security constraint, the optimal uplink and downlink transmission time can be expressed as
	\begin{align}
		t^*_{\text{u}} = \frac{\frac{1}{\rho R_{\text{u}}\left(p_{\text{u}},B_{\text{u}} \right)} D_{\text{th}}}{\frac{R_{\text{SE}}\left(p_{\text{u}},B_{\text{u}} \right) }{R_{\text{u}}\left(p_{\text{u}},B_{\text{u}} \right)} +\frac{  R_{\text{CE}}\left(p_{\text{d}},B_{\text{d}} \right)}{R_{\text{\text{d}}}\left(p_{\text{d}},B_{\text{d}} \right)} }, \label{prop1_1}\\
		t^*_{\text{d}} = \frac{\frac{1}{ R_{\text{d}}\left(p_{\text{d}},B_{\text{d}} \right)} D_{\text{th}}}{\frac{R_{\text{SE}}\left(p_{\text{u}},B_{\text{u}} \right) }{R_{\text{u}}\left(p_{\text{u}},B_{\text{u}} \right)} +\frac{  R_{\text{CE}}\left(p_{\text{d}},B_{\text{d}} \right)}{R_{\text{\text{d}}}\left(p_{\text{d}},B_{\text{d}} \right)} }\label{prop1_2}.
	\end{align}
	Accordingly, the optimization problem \eqref{P1} can be equivalently reformulated as
\begin{subequations}\label{P2}
\begin{align} \min\limits_{\substack{p_{\text{u}},B_{\text{u}},p_{\text{d}},B_{\text{d}}}} &{\frac{R_{\text{SE}}\left(p_{\text{u}},B_{\text{u}} \right) }{R_{\text{u}}\left(p_{\text{u}},B_{\text{u}} \right)} +\frac{  R_{\text{CE}}\left(p_{\text{d}},B_{\text{d}} \right)}{R_{\text{\text{d}}}\left(p_{\text{d}},B_{\text{d}} \right)} }  \\
\begin{split}\text{s.t.} \ &T\left({\frac{R_{\text{SE}}\left(p_{\text{u}},B_{\text{u}} \right) }{R_{\text{u}}\left(p_{\text{u}},B_{\text{u}} \right)} +\frac{  R_{\text{CE}}\left(p_{\text{d}},B_{\text{d}} \right)}{R_{\text{\text{d}}}\left(p_{\text{d}},B_{\text{d}} \right)} } \right) \geq \\ &\quad \quad D_{\text{th}}\left( \frac{1}{\rho R_{\text{u}}\left(p_{\text{u}},B_{\text{u}} \right)}+\frac{1}{R_{\text{d}}\left(p_{\text{d}},B_{\text{d}} \right)}+\frac{\alpha}{ \rho f_{\max}}\right)\end{split} \label{P2b} \\
&0 \leq B_{\text{u}}\leq B_{\max},0 \leq B_{\text{d}}\leq B_{\max},\\&0 \leq p_{\text{u}}\leq P_{\text{umax}}, 0 \leq p_{\text{d}}\leq P_{\text{dmax}}.
\end{align}
\end{subequations}
\end{proposition}
\begin{proof}
	See Appendix \ref{Appendix_Prop1}.
\end{proof}

Based on \textbf{Proposition} \ref{Prop1}, the original optimization problem can be simplified to \eqref{P2}, which focuses on optimizing the bandwidth and power allocation in the $\textbf{SC}^3$ loop. It is noteworthy that the terms $\frac{R_{\text{SE}}\left(p_{\text{u}},B_{\text{u}} \right) }{R_{\text{u}}\left(p_{\text{u}},B_{\text{u}} \right)}$ and $\frac{  R_{\text{CE}}\left(p_{\text{d}},B_{\text{d}} \right)}{R_{\text{\text{d}}}\left(p_{\text{d}},B_{\text{d}} \right)}$ appear in both the objective function and constraint of \eqref{P2}. Analyzing the behavior of these functions is helpful for solving \eqref{P2}. The following lemma reveals its monotonicity.

\begin{lemma}\label{Lemma2}
	For function $f\left( x \right)  \triangleq \frac{\log\left( 1+ax\right) }{\log\left( {1+bx}\right)}$ defined for $x\in\left( 0,+\infty \right) $, with parameters $a>0$ and $b>0$, if $a>b$, $f\left( x \right)$ is decreasing with respect to $x$, otherwise, if $a<b$, $f\left( x \right)$ is increasing with respect to $x$.
\end{lemma}
\begin{proof}
	The derivative of $f\left( x \right)$ can be calculated as
	\begin{equation}\label{lemma2_eq1}
		f'\left( x \right) = \frac{a\left( 1+bx\right) \log \left( 1+bx\right) -b\left( 1+ax\right) \log \left( 1+ax\right) }{\left( 1+ax\right) \left( 1+bx\right) \log^2 \left( 1+bx\right) }.
	\end{equation}
	
	Defining the numerator of $	f'\left( x \right)$ as $g\left( x \right) = a\left( 1+bx\right) \log \left( 1+bx\right) -b\left( 1+ax\right) \log \left( 1+ax\right)$, we have
	\begin{equation}\label{lemma2_eq2}
		g'\left( x \right) = ab\left[ \log \left( 1+bx \right)  - \log \left( 1+ax \right)  \right] .
	\end{equation}
	If $a>b$, we have $g'\left( x \right)<0$ and $g\left( x \right)$ is decreasing. As $g\left( 0 \right) = 0$, we can obtain that $f'\left( x \right)<0$, indicating that $f\left( x \right)$ is decreasing. If $a<b$, it can be proven that $f\left( x \right)$ is increasing in a similar way, which completes the proof.
\end{proof}

According to \textbf{Lemma \ref{Lemma2}}, the monotonicity of $\frac{R_{\text{SE}}\left(p_{\text{u}},B_{\text{u}} \right) }{R_{\text{u}}\left(p_{\text{u}},B_{\text{u}} \right)}$ and $\frac{  R_{\text{CE}}\left(p_{\text{d}},B_{\text{d}} \right)}{R_{\text{\text{d}}}\left(p_{\text{d}},B_{\text{d}} \right)}$ depends on the relative strengths of the channel gains of the legitimate link and its corresponding eavesdropping link. Specifically, if $g_{\text{u}}>g_{\text{SE}}$, $\frac{R_{\text{SE}}\left(p_{\text{u}},B_{\text{u}} \right) }{R_{\text{u}}\left(p_{\text{u}},B_{\text{u}} \right)}$ is increasing with respect to $p_{\text{u}}$ and decreasing with respect to $B_{\text{u}}$, and conversely, if $g_{\text{u}}<g_{\text{SE}}$, the behavior is reversed. Similarly, $\frac{R_{\text{CE}}\left(p_{\text{d}},B_{\text{d}} \right) }{R_{\text{d}}\left(p_{\text{d}},B_{\text{d}} \right)}$ increases with respect to $p_{\text{d}}$ and decreases with respect to $B_{\text{d}}$ if $g_{\text{d}}>g_{\text{CE}}$, and vice versa.

Therefore, if $g_{\text{u}}>g_{\text{SE}}$, to minimize the objective function of \eqref{P2}, $p_{\text{u}}$ should be minimized and $B_{\text{u}}$ should be maximized. Conversely, if $g_{\text{u}}<g_{\text{SE}}$, $p_{\text{u}}$ should be maximized, and $B_{\text{u}}$ should be minimized. The determination of optimal downlink power and bandwidth follows similar reasoning. These observations motivate the following proposition aimed at further simplifying problem \eqref{P2}.

\begin{proposition}\label{Prop2}
	If $g_{\text{u}}>g_{\text{SE}}$, the optimal uplink bandwidth is $B_{\text{u}}^* = B_{\text{max}}$; conversely, if $g_{\text{u}}<g_{\text{SE}}$, the optimal uplink transmit power is $p_{\text{u}}^* = P_{\text{umax}}$. Similarly, for the downlink, we have $B_{\text{d}}^* = B_{\text{max}}$ if $g_{\text{d}}>g_{\text{CE}}$, and $p_{\text{d}}^* = P_{\text{dmax}}$ if $g_{\text{d}}<g_{\text{CE}}$. In addition, for all channel conditions, \eqref{P2b} must hold with equality to minimize the objective function of \eqref{P2}.
\end{proposition}

\begin{proof}
	See Appendix \ref{Appendix_Prop2}.
\end{proof}

\begin{remark}
	The equality in \eqref{P2b} indicates that the available time resources will be fully utilized in the optimal solution. Based on \textbf{Lemma \ref{Lemma1}} and \textbf{Proposition \ref{Prop2}}, we can characterize the utilization of the four types of system resources: time, computing capability, bandwidth, and power. Specifically, the time and computing capability resources are always fully utilized. However, the transmit power and bandwidth resources are more complicated. \textbf{Proposition \ref{Prop2}} indicates that, for each link, one of the two resources will be set to its maximum available value, depending on the relative channel conditions between the legitimate link and its corresponding eavesdropping link. Therefore, among the four optimization variables in problem \eqref{P2}, two will take their maximum values, while the other two may not due to the closed-loop security constraint.
\end{remark}

According to \textbf{Proposition \ref{Prop2}}, there are four distinct scenarios for the optimization problem \eqref{P2}, depending on the channel conditions. In the next section, we propose algorithms to solve the problem \eqref{P2} in different cases.

\section{Algorithms to solve Problem \eqref{P2}}
\label{Section:Algorithm}
This section details the proposed algorithms for solving \eqref{P2}. First, we consider the typical superior legitimate channel case where $g_{\text{u}}>g_{\text{SE}}$ and $g_{\text{d}}>g_{\text{CE}}$, which indicates that the legitimate channel conditions are better than their corresponding eavesdropping channel conditions. For this typical case, we analyze the properties of the optimal solution and propose an efficient algorithm to obtain the global solution. Then, we propose a general algorithm for the other cases.

\subsection{Superior Legitimate Channel Case}
We define the superior legitimate channel case as the scenario where $g_{\text{u}}>g_{\text{SE}}$ and $g_{\text{d}}>g_{\text{CE}}$. In this case, according to \textbf{Proposition \ref{Prop2}}, we have $B_{\text{u}}^* = B_{\text{max}}$ and $B_{\text{d}}^* = B_{\text{max}}$. Therefore, the problem \eqref{P2} can be simplified to the following transmit power optimization problem
	\begin{subequations}\label{P3}
	\begin{align} \min\limits_{\substack{p_{\text{u}},p_{\text{d}}}} \ & {\frac{R_{\text{SE}}\left(p_{\text{u}} \right) }{R_{\text{u}}\left(p_{\text{u}} \right)} +\frac{  R_{\text{CE}}\left(p_{\text{d}} \right)}{R_{\text{\text{d}}}\left(p_{\text{d}} \right)} } \label{P3a} \\
		\begin{split}\text{s.t.} \ &T\left({\frac{R_{\text{SE}}\left(p_{\text{u}} \right) }{R_{\text{u}}\left(p_{\text{u}} \right)} +\frac{  R_{\text{CE}}\left(p_{\text{d}} \right)}{R_{\text{\text{d}}}\left(p_{\text{d}} \right)} } \right) \geq \\&\quad\quad D_{\text{th}}\left( \frac{1}{\rho R_{\text{u}}\left(p_{\text{u}} \right)}+\frac{1}{R_{\text{d}}\left(p_{\text{d}} \right)}+\frac{\alpha}{ \rho f_{\max}}\right),\end{split} \label{P3b} \\
		&0 \leq p_{\text{u}}\leq P_{\text{umax}}, 0 \leq p_{\text{d}}\leq P_{\text{dmax}},
			\end{align}
	\end{subequations}
	where
	\begin{align}
		&R_{\text{u/SE}}\left( p_{\text{u}}\right)  = B_{\text{max}}\log_2(1+\frac{p_{\text{u}}g_{\text{u/SE}}}{B_{\text{max}}N_0}),\\
		&R_{\text{d/CE}}\left( p_{\text{d}}\right)  = B_{\text{max}}\log_2(1+\frac{p_{\text{d}}g_{\text{d/CE}}}{B_{\text{max}}N_0}).
	\end{align}

Although the objective function of \eqref{P3} is non-convex, KKT conditions are still necessary conditions for a locally optimal solution, provided the linear independence constraint qualification (LICQ) holds~\cite[Theorem 6.35]{KKT}. To simplify the presentation of the conditions, we define the following auxiliary functions
\begin{align}
	&f_1\left( p_{\text{u}} \right) = \frac{R_{\text{SE}}\left(p_{\text{u}} \right) }{R_{\text{u}}\left(p_{\text{u}} \right)},\quad
	f_2\left( p_{\text{d}} \right) = \frac{  R_{\text{CE}}\left(p_{\text{d}} \right)}{R_{\text{\text{d}}}\left(p_{\text{d}} \right)},\\
	&h_1 \left( p_{\text{u}} \right) = \frac{D_{\text{th}}}{\rho T R_{\text{u}}\left(p_{\text{u}} \right)},\quad
	h_2 \left( p_{\text{d}} \right) = \frac{D_{\text{th}}}{ T R_{\text{d}}\left(p_{\text{d}} \right)}.
\end{align}

Based on the above definitions, the objective function \eqref{P3a} can be expressed as $f_1\left( p_{\text{u}} \right) + f_2\left( p_{\text{d}} \right) $, and the constraint \eqref{P3b} can be reformulated as $f_1 \left( p_{\text{u}} \right) +  f_2 \left( p_{\text{d}} \right)- h_1 \left( p_{\text{u}} \right)-h_2 \left( p_{\text{d}} \right)\geq  \frac{\alpha D_{\text{th}}}{ \rho T f_{\max}}$. The KKT conditions of \eqref{P3} are provided as follows.
\begin{align}
	&f'_1 \left( p_{\text{u}} \right)-w_1\left( 	f'_1 \left( p_{\text{u}} \right) - 	h'_1 \left( p_{\text{u}} \right) \right)-w_2 =0,\label{KKT_a}\\
	&f'_2 \left( p_{\text{d}} \right)-w_1\left( 	f'_2 \left( p_{\text{d}} \right) - 	h'_2 \left( p_{\text{d}} \right) \right)-w_3 =0,\label{KKT_b}\\
	&w_1 \left[ f_1 \left( p_{\text{u}} \right) +  f_2 \left( p_{\text{d}} \right)- h_1 \left( p_{\text{u}} \right)-h_2 \left( p_{\text{d}} \right)- \frac{\alpha D_{\text{th}}}{ \rho T f_{\max}}\right] =0,\label{KKT_c}\\
	&w_2 \left( P_{\text{umax}}-p_{\text{u}}\right) =0,\label{KKT_d}\\
	&w_3 \left( P_{\text{dmax}}-p_{\text{d}}\right) =0,\label{KKT_e}\\
	&w_1\geq 0, w_2\geq 0, w_3\geq 0,\label{KKT_f}
\end{align}
where $w_1$, $w_2$, and $w_3$ are Lagrange multipliers.

Next, we analyze KKT points based on the above KKT conditions. First, if $p_{\text{u}}<P_{\text{umax}}$ and $p_{\text{d}}<P_{\text{dmax}}$, according to the complementary slackness conditions \eqref{KKT_d} and \eqref{KKT_e}, the corresponding Lagrange multipliers $w_2$ and $w_3$ should be zero. Substituting $w_2 = 0$ and $w_3 = 0$ into \eqref{KKT_a} and \eqref{KKT_b} and simplifying the expressions, we can obtain the following relationship as
\begin{align}\label{equation1}
	\frac{f'_1 \left( p_{\text{u}} \right)}{h'_1 \left( p_{\text{u}} \right)} =\frac{f'_2 \left( p_{\text{d}} \right)}{h'_2 \left( p_{\text{d}} \right)}.
\end{align}
Equation \eqref{equation1} implicitly defines a curve in the $\left( p_{\text{u}}, p_{\text{d}} \right)$ plane, whose characteristics are provided in the following lemma.

\begin{lemma}\label{Lemma3}
	The functions $\frac{f'_1 \left( p_{\text{u}} \right)}{h'_1 \left( p_{\text{u}} \right)}$ and $\frac{f'_2 \left( p_{\text{d}} \right)}{h'_2 \left( p_{\text{d}} \right)}$ are monotonically decreasing with respect to $ p_{\text{u}}$ and $ p_{\text{d}}$, respectively. Therefore, the curve corresponding to \eqref{equation1} is strictly increasing, i.e., for any pair $\left( p_{\text{u}}, p_{\text{d}} \right)$ satisfying \eqref{equation1}, an increase in $p_{\text{u}}$ necessitates a corresponding increase in $ p_{\text{d}}$ to maintain the equality.
\end{lemma}

\begin{proof}
	See Appendix \ref{Appendix_Lemma3}.
\end{proof}

In addition, according to \textbf{Proposition \ref{Prop2}}, the equality of \eqref{P3b} must hold, i.e.,
\begin{align}\label{equation2}
f_1 \left( p_{\text{u}} \right) +  f_2 \left( p_{\text{d}} \right)- h_1 \left( p_{\text{u}} \right)-h_2 \left( p_{\text{d}} \right)=  \frac{\alpha D_{\text{th}}}{ \rho T f_{\max}}.
\end{align}

By solving \eqref{equation1} and \eqref{equation2}, we can obtain a candidate KKT point when $p_{\text{u}}<P_{\text{umax}}$ and $p_{\text{d}}<P_{\text{dmax}}$. Similarly, by substituting $p_{\text{u}}=P_{\text{umax}}$ and $p_{\text{d}}=P_{\text{dmax}}$ into \eqref{equation2}, respectively, we can obtain the other two candidate KKT points. As KKT conditions are necessary for local optimality, the global optimum, which must also be a local optimum, is one of the above three KKT points\footnote{Although the simplified problem \eqref{P3} is non-convex, we can verify that the LICQ holds for \eqref{P3}. Specifically, for any point in the feasible set, there are at most two constraints holding with equality, and their gradients are linearly independent. Therefore, according to \cite{KKT}, any local optimum of \eqref{P3} must satisfy the KKT conditions. Since the global optimum is also a local optimum, it must also satisfy the KKT conditions. Consequently, the global optimum must be one of the KKT points.}. Based on the above analysis, we have the following proposition.
\begin{proposition}\label{Prop3}
	The system of simultaneous equations defined by \eqref{equation1} and \eqref{equation2} has at most one solution in the region where $p_{\text{u}}>0$ and $p_{\text{d}}>0$. Denoting the solution (if it exists) as $\left( \hat{p_{\text{u}}}, \hat{p_{\text{d}}} \right)$, if $\hat{p_{\text{u}}}\leq P_{\text{umax}}$ and $\hat{p_{\text{d}}}\leq P_{\text{dmax}}$, then $\left( \hat{p_{\text{u}}}, \hat{p_{\text{d}}} \right)$ is the globally optimal solution to the optimization problem \eqref{P3}. Otherwise, if the solution does not exist, or if it exists but $\hat{p_{\text{u}}}>P_{\text{umax}}$ or $\hat{p_{\text{d}}}>P_{\text{dmax}}$, we can substitute $p_{\text{u}}=P_{\text{umax}}$ or $p_{\text{d}}=P_{\text{dmax}}$ into \eqref{equation2} and take the solution point with a smaller objective function value as the globally optimal solution.
\end{proposition}

\begin{proof}
	See Appendix \ref{Appendix_Prop3}.
\end{proof}

Based on the above analysis, to find the global solution to \eqref{P3}, we can first find the intersection point of two curves defined by \eqref{equation1} and \eqref{equation2}. As the two curves are monotonically increasing and decreasing, respectively, their intersection point (if exists) is unique and can be obtained efficiently through binary search. If the intersection point $\left( \hat{p_{\text{u}}}, \hat{p_{\text{d}}} \right)$ satisfies the constraints $\hat{p_{\text{u}}}\leq P_{\text{umax}}$ and $\hat{p_{\text{d}}}\leq P_{\text{dmax}}$, it is globally optimal. Otherwise, the optimum lies on the boundary. We can obtain the boundary points by substituting $p_{\text{u}}=P_{\text{umax}}$ and $p_{\text{d}}=P_{\text{dmax}}$ into \eqref{equation2}, and select the point with a smaller objective function value as the globally optimal solution. The algorithm to find the global solution to \eqref{P3} is summarized in \textbf{Algorithm \ref{Algo1}}. In \textbf{Algorithm \ref{Algo1}}, the KKT point $(\hat p_{\text{u}},\hat p_{\text{d}})$ associated with \eqref{equation1} and \eqref{equation2} is obtained via a nested bisection method, whose complexity is $\mathcal{O}(\log^2(1/\epsilon))$ to achieve an accuracy $\epsilon$. If the point is infeasible, two additional scalar equations on the boundary (with $p_{\text{u}}=P_{\text{umax}}$ or $p_{\text{d}}=P_{\text{dmax}}$) are solved by single binary searches, each with complexity $\mathcal{O}(\log(1/\epsilon))$. Therefore, the overall complexity of \textbf{Algorithm \ref{Algo1}} is $\mathcal{O}(\log^2(1/\epsilon))$.

\begin{algorithm}[t]\label{Algo1}
	\caption{The proposed algorithm for solving problem \eqref{P3}}
	\SetKwInOut{Input}{Input}\SetKwInOut{Output}{Output}\SetKwInOut{Initialize}{Initialization}
	\textbf{Input:} System parameter including $P_{\text{umax}}$, $P_{\text{dmax}}$, $D_{\text{th}}$, etc; the convergence tolerance $\epsilon$.\\
	Solve the system of simultaneous equations \eqref{equation1} and \eqref{equation2} through binary search and denote the solution as $\left( \hat{p_{\text{u}}}, \hat{p_{\text{d}}} \right)$ if it exists.\\
	\eIf{The solution exists with $\hat{p_{\text{u}}}\leq P_{\text{umax}}$ and $\hat{p_{\text{d}}}\leq P_{\text{dmax}}$}
	{Set the global solution as $p^*_{\text{u}}=\hat{p_{\text{u}}}$ and $p^*_{\text{d}} = \hat{p_{\text{d}}}$\;}
	{Substitute $p_{\text{u}}=P_{\text{umax}}$ into \eqref{equation2}, solve for $p_{\text{d}}$, and denote the solution as $\overline{p_{\text{d}}}$. Calculate the corresponding objective function value $O_1$\;
	Substitute $p_{\text{d}}=P_{\text{dmax}}$ into \eqref{equation2}, solve for $p_{\text{u}}$, and denote the solution as $\overline{p_{\text{u}}}$. Calculate the corresponding objective function value $O_2$\;	
	\eIf{$O_1<O_2$}{Set the globally optimal solution as $p^*_{\text{u}}=P_{\text{umax}}$ and $p^*_{\text{d}} = \overline{p_{\text{d}}}$\;}
	{Set the globally optimal solution as $p^*_{\text{u}}= \overline{p_{\text{u}}}$ and $p^*_{\text{d}} =P_{\text{dmax}}$\;}}
	\textbf{Output:} the optimal power allocation $\left( p^*_{\text{u}},p^*_{\text{d}}\right) $.
\end{algorithm}

\subsection{Other Channel Cases}
If the superior legitimate channel conditions ($g_{\text{u}}>g_{\text{SE}}$ and $g_{\text{d}}>g_{\text{CE}}$) are not satisfied, although the KKT conditions can still be established, the monotonicity properties in \textbf{Lemma \ref{Lemma3}} may no longer hold and the intersection of \eqref{equation1} and \eqref{equation2} is difficult to obtain. Therefore, \textbf{Algorithm \ref{Algo1}} is not directly applicable. In this subsection, we propose a general algorithm based on the MO theory to address such cases. 

\begin{figure} [t]
	\centering
	\includegraphics[width=\linewidth]{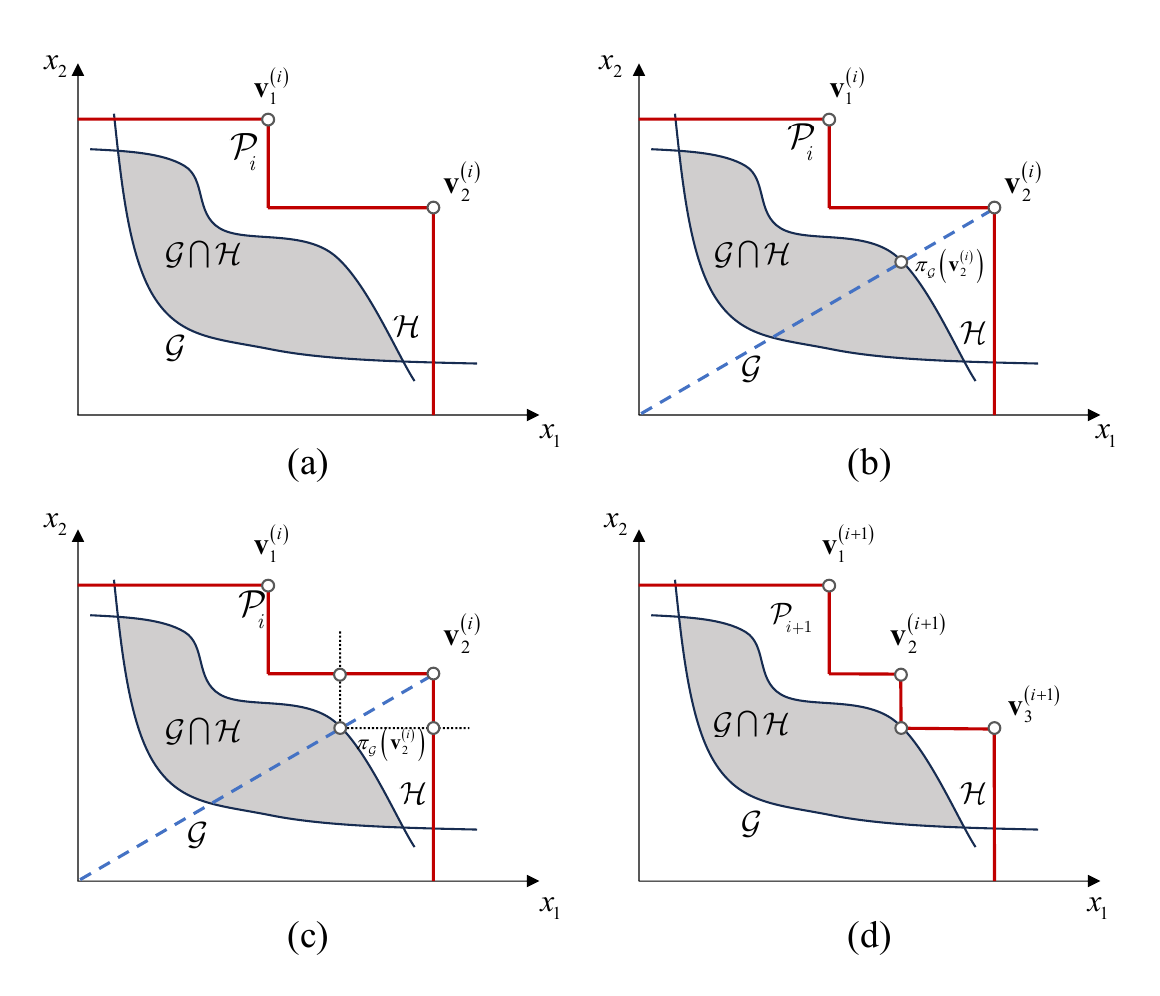}
	\caption{Geometric illustration of the outer polyblock approximation algorithm procedure for a two-dimensional MO problem.}
	\label{fig:mon_optimization}
\end{figure} 

MO is a general global optimization framework for problems whose objective and constraints are monotonic (increasing or decreasing) functions of the decision variables, without requiring convexity, and has been successfully applied to various wireless resource-allocation problems~\cite{mono_optimization1, mono_optimization2, mono_optimization3}. Specifically, MO is concerned with problems of the form~\cite{mono_optimization1}
\begin{align}\label{mop}
	\max \left\{f\left( \mathbf{x} \right)| \mathbf{x} \in \mathcal{G} \cap \mathcal{H} \right\},
\end{align}
where $f\left( \mathbf{x} \right): \mathbb{R}^n \rightarrow \mathbb{R}$ is an increasing function, and $\mathcal{G} \subseteq \left[ \mathbf{0}, \mathbf{b} \right] $ and  $\mathcal{H} \subseteq \left[ \mathbf{0}, \mathbf{b} \right] $ are two sets with the following forms
\begin{align}
	&\mathcal{G} = \left\{ \mathbf{x} \in \left[ \mathbf{0}, \mathbf{b} \right] | \widetilde{g}_i\left( \mathbf{x} \right) \leq 0, i = 0,1,\cdots,m_1   \right\},\\
	&\mathcal{H} = \left\{ \mathbf{x} \in \left[ \mathbf{0}, \mathbf{b} \right] | \widetilde{h}_i\left( \mathbf{x} \right) \geq 0, i = 0,1,\cdots,m_2   \right\},
\end{align}
where $\widetilde{g}_i$ and $\widetilde{h}_i$ are increasing functions of $\mathbf{x}$.

For the MO problem in \eqref{mop}, the optimal solution lies on the upper boundary of the feasible set $\mathcal{G} \cap \mathcal{H}$. Hence, it is sufficient to search along this upper boundary rather than over the entire feasible region, which significantly reduces the computational complexity. A common approach to solving such problems is the polyblock outer approximation algorithm, which constructs a sequence of polyblocks that outer-approximate the feasible set. Specifically, a polyblock $\mathcal{P}$ is defined as a union of boxes $\mathcal{B}_k = [\mathbf{0},\mathbf{v}_k] = \{\mathbf{x}\mid \mathbf{0}\leq \mathbf{x}\leq \mathbf{v}_k\}$, where $\mathbf{v}_k\in\mathcal{V}$ is a vertex of the polyblock and $\mathcal{V}$ is the vertex set. 

Fig.~\ref{fig:mon_optimization} illustrates the outer polyblock approximation algorithm for a two-dimensional MO problem. At the $i$-th iteration, the algorithm evaluates the objective function values at each vertex in the vertex set $\mathcal{V}_i$, and selects a vertex $\mathbf{v}_k$ that yields the maximal value. Next, the algorithm projects it onto the set $\mathcal{G}$, yielding $\pi_{\mathcal{G}}(\mathbf{v}_k)$, as shown in Fig.~\ref{fig:mon_optimization}(b). Then, the cone
$\mathcal{K}^{+}_{\pi_{\mathcal{G}}(\mathbf{v}_k)} \triangleq \{\mathbf{x}\mid \mathbf{x} > \pi_{\mathcal{G}}(\mathbf{v}_k)\}$
is removed from the current polyblock $\mathcal{P}_i$ to obtain a new polyblock $\mathcal{P}_{i+1} = \mathcal{P}_i\setminus\mathcal{K}^{+}_{\pi_{\mathcal{G}}(\mathbf{v}_k)}$ that more tightly outer-approximates the feasible set. Repeating this procedure generates a sequence of polyblocks that increasingly tighten the outer approximation,
\begin{align}
	\mathcal{P}_1 \supset \mathcal{P}_2 \supset \cdots \supset \mathcal{P}_k \supset \cdots \supset \mathcal{G} \cap \mathcal{H}.
\end{align}

During the iterations, the upper bound of the optimal objective value can be obtained as the maximum objective value over all vertices of $\mathcal{P}_i$, denoted as $U^{(i)}$. The lower bound is given by the best objective value among the projection points $\pi_{\mathcal{G}}(\mathbf{v}_k)$, denoted as $L^{(i)}$. It can be shown that $U^{(i)}$ is non-increasing and $L^{(i)}$ is non-decreasing over the iterations. The algorithm terminates when the gap between these two bounds is smaller than a tolerance $\epsilon>0$, i.e., $U^{(i)}-L^{(i)}\le\epsilon$.

According to \cite[Proposition 3.9]{mono_optimization1}, if $f\left( \mathbf{x} \right)$ is Lipschitz-continuous, then the above outer polyblock approximation algorithm is guaranteed to converge to an $\epsilon$-optimal solution in a finite number of iterations for any given $\epsilon>0$.  The detailed proof of the above proposition can be found in \cite{mono_optimization4}.

We now apply the above MO framework to the considered resource
allocation problem via variable substitutions. To illustrate the procedure, we take the superior eavesdropping channel case where $g_{\text{u}}<g_{\text{SE}}$ and $g_{\text{d}}<g_{\text{CE}}$ as an example. In this case, we have $p_{\text{u}}^* = P_{\text{umax}}$ and $p_{\text{d}}^* = P_{\text{dmax}}$ according to \textbf{Proposition \ref{Prop2}}. The problem \eqref{P2} can be simplified as
	\begin{subequations}\label{P4}
	\begin{align} \min\limits_{\substack{B_{\text{u}},B_{\text{d}}}} \ &{\frac{R_{\text{SE}}\left(P_{\text{umax}},B_{\text{u}} \right) }{R_{\text{u}}\left(P_{\text{umax}},B_{\text{u}} \right)} +\frac{  R_{\text{CE}}\left(P_{\text{dmax}},B_{\text{d}} \right)}{R_{\text{\text{d}}}\left(P_{\text{dmax}},B_{\text{d}} \right)} } \label{P4a} \\
		\begin{split}\text{s.t.} \ &T\left({\frac{R_{\text{SE}}\left(P_{\text{umax}},B_{\text{u}} \right) }{R_{\text{u}}\left(P_{\text{umax}},B_{\text{u}} \right)} +\frac{  R_{\text{CE}}\left(P_{\text{dmax}},B_{\text{d}} \right)}{R_{\text{\text{d}}}\left(P_{\text{dmax}},B_{\text{d}} \right)} } \right) \geq \\&\quad D_{\text{th}}\left( \frac{1}{\rho R_{\text{u}}\left(P_{\text{umax}},B_{\text{u}} \right)}+\frac{1}{R_{\text{d}}\left(P_{\text{dmax}},B_{\text{d}} \right)}+\frac{\alpha}{ \rho f_{\max}}\right),\end{split} \label{P4b} \\
		&0 \leq B_{\text{u}}\leq B_{\text{max}},0 \leq B_{\text{d}}\leq B_{\text{max}}.
	\end{align}
\end{subequations}

By substituting $B_{\text{u}} = \frac{1}{x_{\text{u}}}$ and $B_{\text{d}} = \frac{1}{x_{\text{d}}}$, the problem \eqref{P4} can be reformulated as
	\begin{subequations}\label{P5}
	\begin{align} \max\limits_{\substack{x_{\text{u}},x_{\text{d}}}} &-\left[ \frac{R_{\text{SE}}\left(P_{\text{umax}},\frac{1}{x_{\text{u}}} \right) }{R_{\text{u}}\left(P_{\text{umax}},\frac{1}{x_{\text{u}}} \right)} +\frac{  R_{\text{CE}}\left(P_{\text{dmax}},\frac{1}{x_{\text{d}}} \right)}{R_{\text{\text{d}}}\left(P_{\text{dmax}},\frac{1}{x_{\text{d}}} \right)}  \right] \label{P5a} \\
		\begin{split}\text{s.t.} \ &D_{\text{th}}\left( \frac{1}{\rho R_{\text{u}}\left(P_{\text{umax}},\frac{1}{x_{\text{u}}} \right)}+\frac{1}{R_{\text{d}}\left(P_{\text{dmax}},\frac{1}{x_{\text{d}}} \right)}+\frac{\alpha}{ \rho f_{\max}}\right)\\&\quad \quad -T\left({\frac{R_{\text{SE}}\left(P_{\text{umax}},\frac{1}{x_{\text{u}}} \right) }{R_{\text{u}}\left(P_{\text{umax}},\frac{1}{x_{\text{u}}} \right)} +\frac{  R_{\text{CE}}\left(P_{\text{dmax}},\frac{1}{x_{\text{d}}} \right)}{R_{\text{\text{d}}}\left(P_{\text{dmax}},\frac{1}{x_{\text{d}}} \right)} } \right) \leq 0,\end{split} \label{P5b} \\
		&x_{\text{u}}-\frac{1}{B_{\text{max}}}\geq 0, x_{\text{d}}-\frac{1}{B_{\text{max}}} \geq 0.
	\end{align}
\end{subequations}

Denoting the objective function of \eqref{P5} as $\widetilde{f}_{\text{obj}}\left( x_{\text{u}},x_{\text{d}}\right)$, we can prove that $\widetilde{f}_{\text{obj}}\left( x_{\text{u}},x_{\text{d}}\right)$ is increasing with respect to $x_{\text{u}}$ and $x_{\text{d}}$ by utilizing \textbf{Lemma \ref{Lemma2}}. In addition, it can be shown that the left side of each constraint is monotonically increasing. Therefore, \eqref{P5} is an MO problem, which can be solved via the polyblock outer approximation algorithm. The detailed algorithm to solve \eqref{P5} is summarized in \textbf{Algorithm \ref{Algo2}}.

For the other two cases (i.e., $g_{\text{u}}>g_{\text{SE}}$ and $g_{\text{d}}<g_{\text{CE}}$, or $g_{\text{u}}<g_{\text{SE}}$ and $g_{\text{d}}>g_{\text{CE}}$), the same MO approach can be applied via a similar reformulation.

\begin{remark}
	At each iteration, the polyblock outer approximation algorithm evaluates the objective function at newly-generated vertices, calculates the projection via a one-dimensional bisection search, and updates the polyblock. The per-iteration complexity is
	$\mathcal{O}\left( N+\log(1/\epsilon_{\rm proj})+|\mathcal{V}_i|\right) $, where $N$ denotes the variable dimension, and $\epsilon_{\rm proj}$ denotes the accuracy of the bisection search. In general, the number of iterations grows exponentially with $N$. Therefore, the overall complexity of the algorithm grows exponentially with the number of variables~\cite{mono_optimization2}. However, the complexity of polyblock outer approximation algorithm is much lower than general global optimization by exploiting the special structure of the problems~\cite{mono_optimization1}. Moreover, in our setting, the dimension of the MO problem in \eqref{P5} is reduced to $N=2$, so that the overall complexity may be acceptable in practice.
\end{remark}

\begin{algorithm}[t]\label{Algo2}
	\caption{The polyblock outer approximation algorithm for solving problem \eqref{P5}}
	\SetKwInOut{Input}{Input}\SetKwInOut{Output}{Output}\SetKwInOut{Initialize}{Initialization}
	\textbf{Input:} System parameter $P_{\text{umax}}$, $P_{\text{dmax}}$, $D_{\text{th}}$, etc; the convergence tolerance $\epsilon$.\\
	Initialize polyblock $\mathcal{P}^{\left( 1\right) }$, characterized by the vertex set $\mathcal{V}^{\left( 1\right) } = \left\{ \mathbf{v}^{\left( 1\right) }_1 \right\}$, where $\mathbf{v}^{\left( 1\right) }_1 = \left(x_{\text{u,max}},x_{\text{d,max}} \right)$, and $x_{\text{u,max}}$ and $x_{\text{d,max}}$ are two values such that the constraint \eqref{P5b} is satisfied with equality at the points $\left(x_{\text{u,max}},1/B_{\text{max}} \right)$ and $\left(1/B_{\text{max}},x_{\text{d,max}} \right)$, respectively. Set $i=1$. Set the current objective function value $f^{\left( 0\right) }=-\infty$\\
	\Repeat{$|\overline{f}^{\left( i\right)}-\widetilde{f}_{\text{obj}}\left(\mathbf{v}^{\left( i\right) }\right) |<\epsilon$}{From $\mathcal{V}^{\left( i\right) }$ find the vertex that maximizes the objective function of \eqref{P5}, denoted as $\mathbf{v}^{\left( i\right) }= \arg \max\limits_{\mathbf{v}^{\left( i\right) }_k \in \mathcal{V}^{\left( i\right) }} \widetilde{f}_{\text{obj}}\left( \mathbf{v}^{\left( i\right) }_k\right)$\;
	Compute the projection of $\mathbf{v}^{\left( i\right) }$ through binary search, denoted as $\mathbf{x}^{\left( i\right) } = \pi_{\mathcal{G}}\left(\mathbf{v}^{\left( i\right) } \right)$\;
	\eIf{$\mathbf{x}^{\left( i\right)}\geq\left( 1/B_{\text{max}},1/B_{\text{max}}\right) $ and $\widetilde{f}_{\text{obj}}\left( \mathbf{x}^{\left( i\right)}\right)>f^{\left( i-1\right) }$}{Set the current best solution and objective function value as $\overline{\mathbf{x}}^{\left( i\right)} = \mathbf{x}^{\left( i\right)}$ and $\overline{f}^{\left( i\right)} = \widetilde{f}_{\text{obj}}\left( \mathbf{x}^{\left( i\right)}\right)$.}{Set $\overline{\mathbf{x}}^{\left( i\right)} = \overline{\mathbf{x}}^{\left( i-1\right)}$ and $\overline{f}^{\left( i\right)} = \overline{f}^{\left( i-1\right)}$.}
	Construct a new polyblock $\mathcal{P}^{\left( i+1\right) }$ with the vertex set $\mathcal{V}^{\left( i+1\right) }$, where $\mathcal{V}^{\left( i+1\right) }$ is generated by $\mathcal{V}^{\left( i+1\right)} = \mathcal{V}^{\left( i\right)}\setminus \mathcal{V}_*\cap \left\{\mathbf{v} \!+\! \left(x^{\left( i\right) }_d \! -\! v_d \right) \mathbf{e}_d|\ \mathbf{v}\in \mathcal{V}_*, d\in\left\{1,2\right\} \right\}$, where $\mathcal{V}_* = \left\{\mathbf{v} \in \mathcal{V}^{\left( i\right)}|\mathbf{v}>\mathbf{x}^{\left( i\right) }\right\}$\;
	Remove from  $\mathcal{V}^{\left( i+1\right)}$ the improper vertices $ \left\{\mathbf{v} \in \mathcal{V}^{\left( i+1\right)}|\mathbf{v}\ngeq\left( 1/B_{\text{max}},1/B_{\text{max}}\right)\right\}$
	Set $i = i+1$\;}
	\textbf{Output:} the optimal solution $\overline{\mathbf{x}}^{\left( i\right)}$ to \eqref{P5}.
\end{algorithm}

\subsection{Summary}
In the considered system, the proposed resource allocation procedures are executed at the EIH.	At the beginning of each configuration period, the EIH updates the positions of the sensor, the robot,	and the eavesdropper using the UAV-mounted camera or radar, and then evaluates the corresponding channel gains based on a channel model or a radio map. If the closed-loop security constraint is not active, the bandwidth and transmit power are set to their maximum values. Otherwise, according to the channel ordering, the EIH applies the case-dependent solution method to obtain the resource allocation, as summarized in Table~\ref{tab2}. The computational complexity of Algorithm~\ref{Algo1} is $\mathcal{O}(\log^2(1/\epsilon))$, which is computationally efficient. Moreover, although the complexity of Algorithm~\ref{Algo2} grows exponentially with the number of variables, problem \eqref{P5} has been reduced to a two-dimensional form,	making the complexity acceptable for practical implementation.

\begin{table*}[t]
	\centering
	\caption{Optimal structure of variables and complexity under different channel conditions.}
	\label{tab2}		
	\begin{tabular}{c|c|c|c|c|c|c|c|c|c}
		\hline
		\multirow{2}{*}{Case} &
		\multicolumn{1}{c|}{\multirow{2}{*}{\makecell{Channel \\condition}}} &
		\multicolumn{7}{c|}{Optimal variable} &
		\multirow{2}{*}{Complexity} \\
		\cline{3-9}
		& & $p_{\text{u}}^*$ & $t_{\text{u}}^*$ & $B_{\text{u}}^*$ & $f^*$ & $p_{\text{d}}^*$ & $t_{\text{d}}^*$ & $B_{\text{d}}^*$ & \\
		\hline
		
		\makecell{I} &
		\makecell{$g_{\text{u}}>g_{\text{SE}}$ \\ $g_{\text{d}}>g_{\text{CE}}$} &
		\makecell{KKT-based \\algorithm} & \eqref{prop1_1} & $B_{\text{max}}$ & $f_{\text{max}}$ & \makecell{KKT-based \\algorithm} & \eqref{prop1_2} & $B_{\text{max}}$ &
		$\mathcal{O}(\log^2(1/\epsilon))$ \\
		\hline
		
		\makecell{II} &
		\makecell{$g_{\text{u}}<g_{\text{SE}}$ \\ $g_{\text{d}}<g_{\text{CE}}$} &
		$P_{\text{umax}}$ & \eqref{prop1_1} & \makecell{MO-based \\algorithm} & $f_{\text{max}}$ & $P_{\text{dmax}}$ & \eqref{prop1_2} & \makecell{MO-based \\algorithm} &
		\makecell{$\mathcal{O}\left( N+\log(1/\epsilon_{\rm proj})+|\mathcal{V}_i|\right) $\\ per iteration} \\
		\hline
		
		III &
		\makecell{$g_{\text{u}}>g_{\text{SE}}$ \\ $g_{\text{d}}<g_{\text{CE}}$} &
		\makecell{MO-based \\algorithm} & \eqref{prop1_1} & $B_{\text{max}}$ & $f_{\text{max}}$ & $P_{\text{dmax}}$ & \eqref{prop1_2} &\makecell{MO-based \\algorithm} &
			\makecell{$\mathcal{O}\left( N+\log(1/\epsilon_{\rm proj})+|\mathcal{V}_i|\right) $\\ per iteration} \\
		\hline
		
		IV &
		\makecell{$g_{\text{u}}<g_{\text{SE}}$ \\ $g_{\text{d}}>g_{\text{CE}}$} &
		$P_{\text{umax}}$ & \eqref{prop1_1} & \makecell{MO-based \\algorithm} & $f_{\text{max}}$ & \makecell{MO-based \\algorithm} & \eqref{prop1_2} & $B_{\text{max}}$ &
			\makecell{$\mathcal{O}\left( N+\log(1/\epsilon_{\rm proj})+|\mathcal{V}_i|\right) $\\ per iteration} \\
		\hline
	\end{tabular}
	\vspace{1mm}
	\\
	{\footnotesize Note: Case I corresponds to the superior legitimate channel case and Case II corresponds to the superior eavesdropping channel case.}
\end{table*}

\section{Simulation Results}
\label{Section:Simulation}
In this section, simulation results are provided to evaluate the proposed algorithm. Unless specified otherwise, the simulation parameters are set as: $P_{\text{umax}} = 1$ W, $P_{\text{dmax}} = 1$ W, $B_{\text{max}} = 20$ kHz, $T = 0.1$ s, $N_0=-174$ dBm/Hz, and $D_{\text{th}} = 300$ bits. The air-to-ground and ground-to-air channel gains are modeled as $g = \kappa_1 d^{-\eta_1}$, while the ground-to-ground channel gains are calculated as $g = \kappa_2 d^{-\eta_2}$, where $\kappa_1$ and $\kappa_2$ denote the channel gain at a reference distance of 1 m, $\eta_1$ and $\eta_2$ are corresponding path loss exponents, and $d$ denotes the distance between the transmitter and receiver. We set $\eta_1=2$ following the free-space path loss model~\cite{channel}, and set $\eta_2= 3.5 $~\cite{channel1}. The reference channel gains are set as: $\kappa_1=\kappa_2 = 1.42\times 10^{-4}$~\cite{channel1}. The computing-related parameters are given by: $\alpha = 200$ cycles/bit, $\rho = 0.01$, and $f_{\text{max}} = 1$ GHz.

In order to evaluate the performance of the proposed scheme, we compare it with the following benchmark schemes.
\begin{itemize}
	\item Control-oriented closed-loop optimization scheme \cite{single_loop}: jointly allocating the bandwidth, power, and transmission time of the uplink and downlink links, as well as the computing capability of the EIH, aiming to maximize the CNE without eavesdropping. If the eavesdropped closed-loop information exceeds the threshold, we reduce the uplink and downlink transmission time in equal proportion to keep the closed-loop security constraint satisfied.
	\item Separate uplink and downlink optimization scheme: Set separate security constraints for the sensor data transmission and control command transmission, and then maximize data throughput for the two links, respectively.
\end{itemize}

It is worth noting that we have also evaluated conventional secrecy-rate-maximization baseline commonly used in communication systems. However, this baseline typically leads to very small CNE, as it ignores the uplink–downlink balance within the $\textbf{SC}^3$ closed loop. Therefore, for clarity of presentation, the curve for the secrecy-rate-maximization scheme is omitted in the following figures.

First, we consider the superior legitimate channel case where $g_{\text{u}}>g_{\text{SE}}$ and $g_{\text{d}}>g_{\text{CE}}$. In such case, the locations of the sensor, the EIH, the robot, and the eavesdropper are set as $\mathbf{l}_{\text{s}} = [-500, 0, 0]$, $\mathbf{l}_{\text{EIH}} = [500, 0, 100]$, $\mathbf{l}_{\text{r}} = [-500, 0, 0]$, and $\mathbf{l}_{\text{e}} = [-500, 1000, 0]$ in meters.
\begin{figure} [t]
	\centering
	\includegraphics[width=0.9\linewidth]{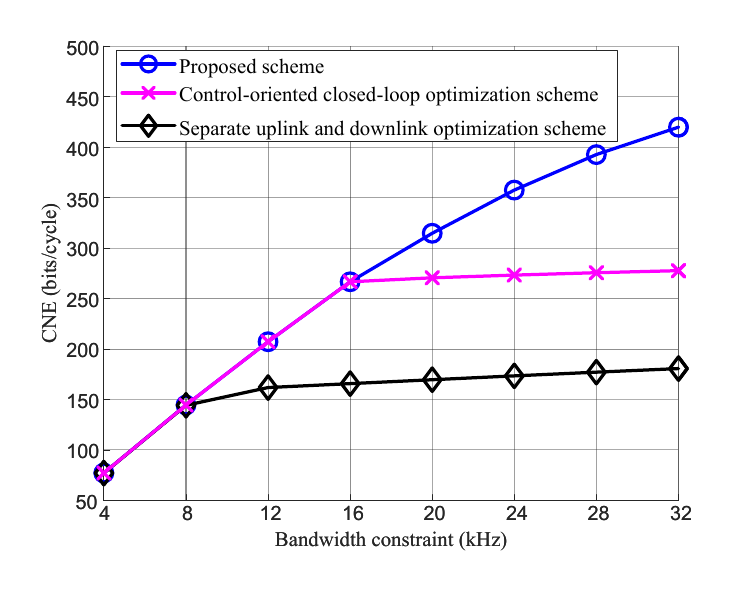}
	\caption{The CNE achieved by different schemes varying with the bandwidth constraint.}
	\label{fig:simu_CNE_vs_B}
\end{figure} 

In Fig. \ref{fig:simu_CNE_vs_B}, we compare the CNE for varying bandwidth constraints. From the figure, it can be seen that when the bandwidth constraint is lower than 16 kHz, the proposed scheme and the control-oriented closed-loop optimization scheme yield identical CNE values. The reason is that the available bandwidth is the primary performance bottleneck in this regime, rendering the security constraint inactive. In such cases, the two schemes utilize the maximum permitted bandwidth and transmit power, achieving identical CNE values. When the bandwidth constraint is higher than 16 kHz, the proposed scheme achieves the highest CNE among the three schemes, which shows the superiority of the proposed scheme. Notably, the CNE attained by the separate uplink and downlink optimization scheme is the lowest. This result underscores the benefit of a comprehensive, joint consideration of both uplink and downlink, as implemented in the proposed scheme.

\begin{figure} [t]
	\centering
	\includegraphics[width=0.9\linewidth]{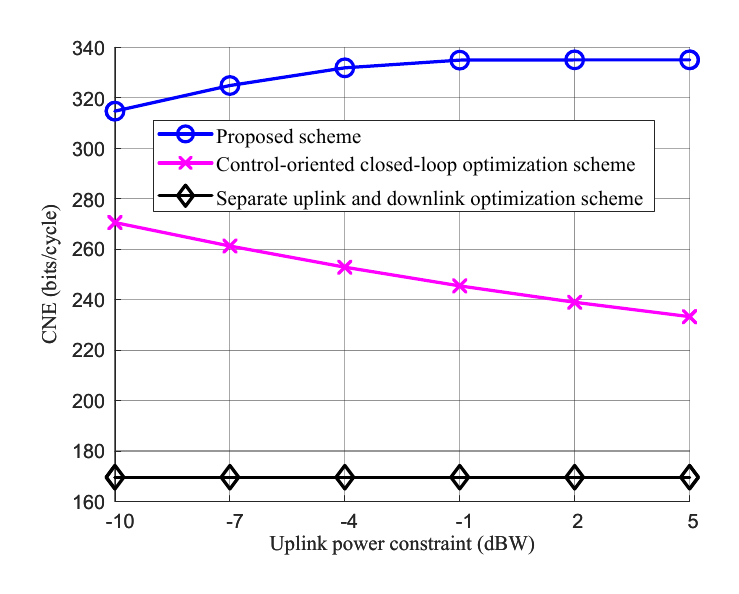}
	\caption{The CNE achieved by different schemes varying with the uplink power constraint.}
	\label{fig:simu_CNE_vs_pu}
\end{figure} 

Fig. \ref{fig:simu_CNE_vs_pu} shows the CNE for different maximum uplink power ($P_{\text{umax}}$). Interestingly, the CNE achieved with three schemes behaves differently with respect to the uplink power constraint. Specifically, the CNE increases with the proposed scheme. This is because the unconstrained optimal operating point of $p_{\text{u}}$, i.e., $\hat{p_{\text{u}}}$ in \textbf{Proposition \ref{Prop3}}, is higher than the maximum uplink power $P_{\text{umax}}$. Increasing $P_{\text{umax}}$ allows the algorithm to select solutions closer to the unconstrained optimal operating point, thereby enhancing CNE. In contrast, the CNE for the separate uplink and downlink optimization scheme remains invariant to changes in $P_{\text{umax}}$. This indicates that the overall CNE, constrained by $\min \left\{ \rho D_{\text{u}},D_{\text{d}}\right\}$, is bottlenecked by the downlink data throughput  $D_{\text{d}}$. Improving the uplink capability for the separate uplink and downlink optimization scheme does not alleviate the downlink bottleneck and consequently yields no improvement in the CNE. In addition, the CNE achieved with the control-oriented closed-loop optimization scheme decreases with the uplink power constraint. This counter-intuitive trend is because this scheme does not consider the closed-loop security constraint. Simply improving the uplink power when $g_{\text{u}}>g_{\text{SE}}$ amplifies the eavesdropped link more than the legitimate link, leading to a decrease of CNE.

\begin{figure} [t]
	\centering
	\includegraphics[width=0.9\linewidth]{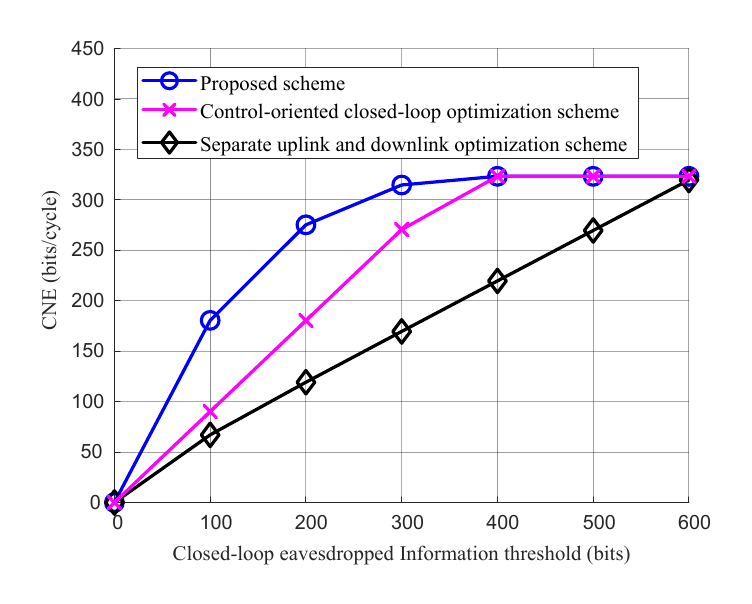}
	\caption{The CNE achieved by different schemes varying with the closed-loop eavesdropped information threshold.}
	\label{fig:simu_CNE_vs_Dth}
\end{figure} 

Fig. \ref{fig:simu_CNE_vs_Dth} presents the achieved CNE versus the closed-loop eavesdropped information thresholds, i.e., $D_{\text{th}}$. It is shown that the CNE achieved with the proposed scheme is highest among the three schemes when $D_{\text{th}}< 400$ bits, showing its superiority. When $D_{\text{th}}\geq 600$ bits, all three schemes achieve the same CNE, which remains constant. This is because when $D_{\text{th}}$ is sufficiently high, the closed-loop security constraint becomes inactive. In such cases, these schemes can reach the maximum CNE by fully utilizing all the system resources. When $D_{\text{th}}\leq 300$ bits, it is also observed that the CNE achieved with the control-oriented closed-loop optimization scheme exhibits an approximately linear dependence on $D_{\text{th}}$. This is because the transmission time is decreased in equal proportion with $\frac{D_{\text{th}}}{\rho D_{\text{SE},0} + D_{\text{CE},0}}$. As the data throughput is in proportion to the transmission time, the achieved CNE is also linearly related to  $D_{\text{th}}$. Comparing the proposed scheme and the control-oriented closed-loop optimization scheme in the security-constrained region ($D_{\text{th}}\leq 400$), we can conclude that the proposed scheme can achieve an internal coordination mechanism such that the CNE can decrease by a smaller proportion if the closed-loop eavesdropped information threshold decreases by a certain proportion.

\begin{figure} [t]
	\centering
	\includegraphics[width=0.9\linewidth]{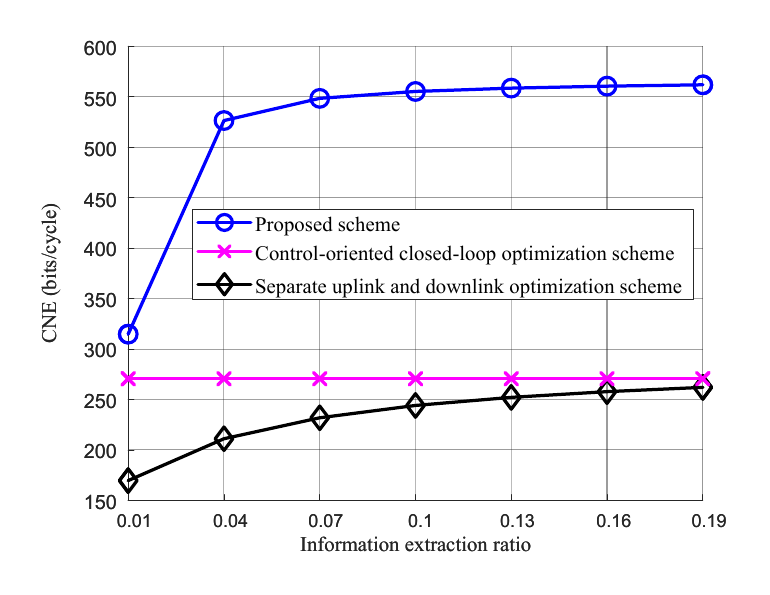}
	\caption{The CNE achieved by different schemes varying with the information extraction ratio.}
	\label{fig:simu_CNE_vs_rho}
\end{figure} 

In Fig. \ref{fig:simu_CNE_vs_rho}, we show the CNE for different information extraction ratios ($\rho$). It can be observed that the control-oriented closed-loop optimization baseline achieves nearly the same CNE for all $\rho$. The reason is that increasing $\rho$ increases both the amount of task-relevant information extracted from the uplink and the information leakage, so the achievable CNE is still limited by the security threshold. In contrast, the proposed scheme benefits from a larger $\rho$ and achieves a higher CNE, since it can adaptively rebalance resources to mitigate the uplink bottleneck.

\begin{figure} [t]
	\centering
	\includegraphics[width=\linewidth]{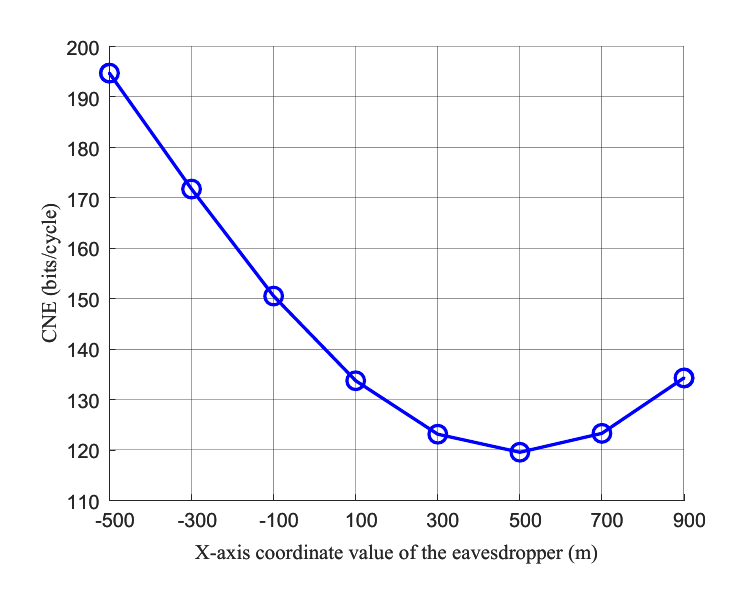}
	\caption{The CNE with the x-axis coordinate values of the eavesdropper.}
	\label{fig:simu_CNE_vs_le}
\end{figure} 

Fig. \ref{fig:simu_CNE_vs_le} illustrates how the CNE is influenced by the location of the eavesdropper. Specifically, we set the position of the eavesdropper as $\mathbf{l}_{\text{e}} = [l_{\text{e,x}}, 1100, 0]$, and show the relationship between the CNE and the x-coordinate $l_{\text{e,x}}$. It can be seen that the CNE first decreases and then increases with $l_{\text{e,x}}$, achieving its minimum value when $l_{\text{e,x}} = -500$ m. As $\mathbf{l}_{\text{s}} = [-500, 0, 0]$ and $\mathbf{l}_{\text{EIH}} = [500, 0, 100]$, it is concluded that the eavesdropper should be closer to the EIH rather than the sensor to degrade the performance of the $\textbf{SC}^3$ closed loop. This is because eavesdropping on the EIH-to-robot link is more effective as the sensor data contains redundant raw data. This figure indicates that a strategic design of the position is a significant consideration from both the perspective of a legitimate system designer aiming to safeguard the system and an eavesdropper seeking to maximize interception effectiveness.

\begin{figure} [t]
	\centering
	\includegraphics[width=0.9\linewidth]{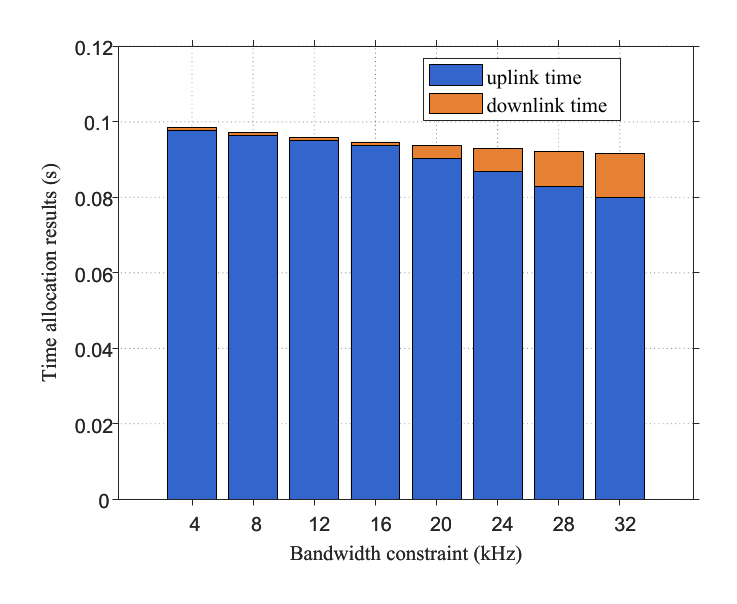}
	\caption{The time allocation results with different bandwidth constraints.}
	\label{fig:simu_T_vs_B}
\end{figure} 

Fig. \ref{fig:simu_T_vs_B} shows the optimal time allocation results for different bandwidth constraints with a fixed total cycle time $T=0.1$ s. It can be seen that the allocated uplink time is much longer than the downlink time. This is because the amount of raw sensing data transmitted on the uplink is substantially greater than the control commands transmitted on the downlink. It can also be observed that the transmission time (the sum of uplink time and downlink time) decreases as the bandwidth constraint increases. The reason is that more data are transmitted and processed with more bandwidth, requiring more computation time and therefore leaving less time for transmission. In addition, the ratio of uplink time to downlink time, $t_{\text{u}}/t_{\text{d}}$, decreases as the bandwidth constraint increases. The reason is analyzed as follows. According to \textbf{Proposition \ref{Prop1}}, the transmission time is inversely proportional to the corresponding data rate, i.e., $t_{\text{u}}/t_{\text{d}} = R_{\text{d}}/\left( \rho R_{\text{u}}\right) $. As the bandwidth increases, both the uplink power and downlink power need to be reduced to satisfy the security constraint. However, the downlink data is more critical for the closed-loop security because the uplink data contains redundant raw data. Therefore, the downlink power decreases more rapidly than the uplink power, leading to an increase in the ratio of uplink to downlink data rates. As a result, the ratio of uplink time to downlink time decreases as the bandwidth increases.

\begin{figure} [t]
	\centering
	\includegraphics[width=0.9\linewidth]{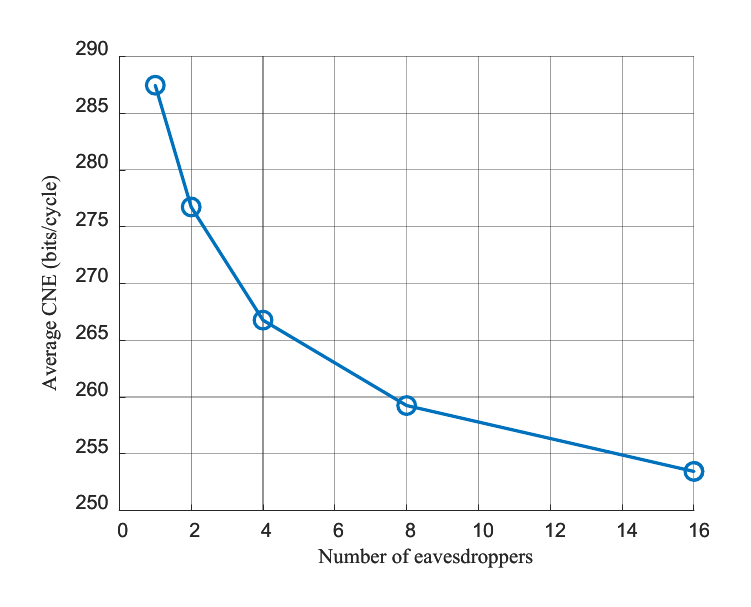}
	\caption{The average CNE varying with the number of eavesdroppers.}
	\label{fig:simu_multiEve}
\end{figure} 

In Fig. \ref{fig:simu_multiEve}, we evaluate the proposed scheme in a multi-eavesdropper scenario. Since the proposed algorithm is derived for the single-eavesdropper case, we extend it by selecting the strongest eavesdropping link on the uplink and the downlink among all eavesdroppers. We randomly generate 1,000 independent eavesdropper locations and show the average CNE versus the number of eavesdroppers $K$. It is shown that the average CNE decreases as $K$ increases, because the system is more likely to be eavesdropped when the number of eavesdroppers increases.

\begin{figure} [t]
	\centering
	\includegraphics[width=0.9\linewidth]{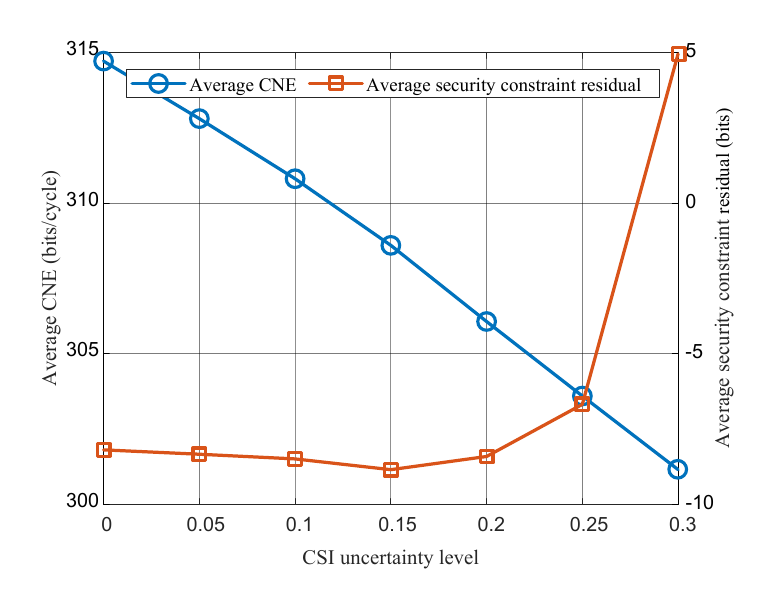}
	\caption{The average CNE and security constraint residual with the CSI uncertainty level.}
	\label{fig:simu_imperfectCSI}
\end{figure} 

In Fig. \ref{fig:simu_imperfectCSI}, we show the impact of the imperfect CSI in practical deployments. Specifically, for each link, we model the estimated large-scale channel gain as $g_{\text{est}}=g_{\text{true}}\left( 1+e\right) $, where $g_{\text{true}}$ is the true channel gain, $e\sim \mathcal{N}\left( 0, \mu^2\right) $ is the relative CSI error, and $\mu$ denotes the CSI uncertainty level. Moreover, small-scale Rayleigh fading is considered for the ground-to-ground sensor-eavesdropper link. The algorithm is executed using the estimated channel gains, while the performance is evaluated using the true channel gains. The results are averaged over 10000 Monte-Carlo trials. As $\mu$ increases, the allocation obtained using imperfect CSI gradually deviates from the true optimum, leading to a decrease in the average CNE. We also report the average security constraint residual, defined as $\mathbb{E}\!\left[\rho D_{\text{SE}} + D_{\text{CE}} - D_{\text{th}}\right]$, where a negative value indicates that the security constraint is satisfied with a safety margin. When $\mu$ is small, the residual can be below zero mainly due to the conservative security analysis in \eqref{RSEa}-\eqref{RSEc}. However, for larger $\mu$, the CSI mismatch dominates, and the residual increases rapidly with the uncertainty level, leading to a higher security risk. This highlights the necessity of incorporating robust designs when the CSI quality is limited by incorporating the CSI error statistics.

Next, we show the simulation results in the superior eavesdropping channel case where $g_{\text{u}}<g_{\text{SE}}$ and $g_{\text{d}}<g_{\text{CE}}$. Specifically, the locations of the sensor, the EIH, the robot, and the eavesdropper are set as $\mathbf{l}_{\text{s}} = [-500, 0, 0]$, $\mathbf{l}_{\text{EIH}} = [500, 0, 100]$, $\mathbf{l}_{\text{r}} = [-800, 0, 0]$, and $\mathbf{l}_{\text{e}} = [-460, 0, 0]$ in meters.

\begin{figure} [t]
	\centering
	\includegraphics[width=0.9\linewidth]{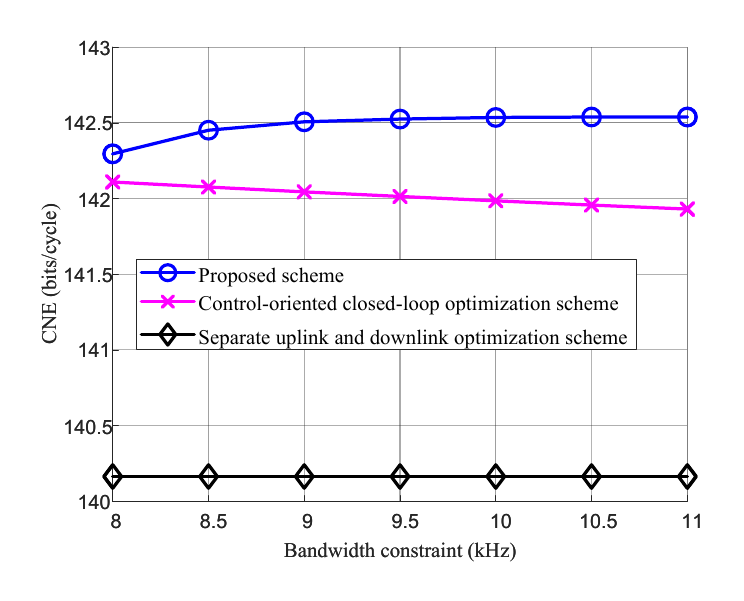}
	\caption{The CNE achieved by different schemes varying with the bandwidth constraint in the superior eavesdropping channel case.}
	\label{fig:simu_CNE_vs_B_case2}
\end{figure} 

Fig. \ref{fig:simu_CNE_vs_B_case2} illustrates the CNE achieved with three schemes under different bandwidth constraints in the superior eavesdropping channel case. It is shown that the trend of CNE with respect to bandwidth constraint in this scenario is similar to the trend of CNE with respect to uplink power constraint $P_{\text{umax}}$ in the superior legitimate channel case (previously shown in Fig. \ref{fig:simu_CNE_vs_pu}). This correspondence can be attributed to \textbf{Lemma \ref{Lemma2}}, which establishes that the monotonicity of the objective function in \eqref{P2} with respect to the bandwidth in the superior eavesdropping channel case is the same as that with respect to the transmit power in the superior legitimate channel case.

\begin{figure} [t]
	\centering
	\includegraphics[width=0.9\linewidth]{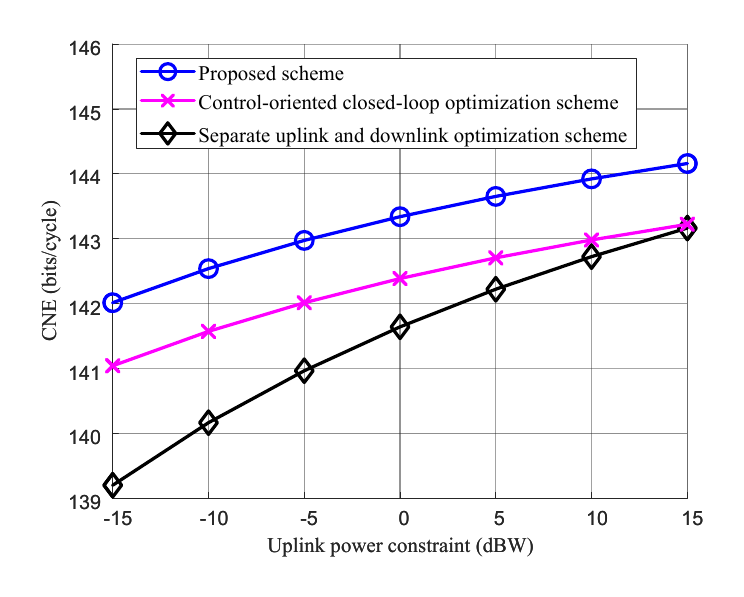}
	\caption{The CNE achieved by different schemes varying with the uplink power constraint in the superior eavesdropping channel case.}
	\label{fig:simu_CNE_vs_pu_case2}
\end{figure}

In Fig. \ref{fig:simu_CNE_vs_pu_case2}, we show the CNE achieved with three schemes under different uplink power constraints in the superior eavesdropping channel case. By comparing Fig. \ref{fig:simu_CNE_vs_pu} and Fig. \ref{fig:simu_CNE_vs_pu_case2}, we can find a significant difference in behavior: in the superior eavesdropping channel case, the CNE increases with $P_{\text{umax}}$ for all three schemes; while in the superior legitimate channel case, the separate optimization scheme's CNE was invariant and the control-oriented scheme's CNE decreases with $P_{\text{umax}}$. The increasing trend for all three schemes in the superior eavesdropping channel case is because improving the uplink power yields a more substantial improvement in the legitimate transmission rate compared to the increase in the eavesdropping rate, as shown in \textbf{Lemma \ref{Lemma2}}. Therefore, higher uplink power enhances the potential for secure communication, leading to improved CNE across all schemes. In addition, it can be observed that the separate scheme exhibits a faster marginal increase than the proposed scheme. The reason is that the separate scheme optimizes the uplink and downlink independently, and thus can benefit from the uplink power increase more significantly before the closed-loop bottleneck is activated. In contrast, the proposed scheme can achieve an uplink-downlink balance. This balanced approach inherently leads to diminishing marginal gains from further uplink power increases. It is worth noting that, despite the slower marginal increase with respect to the uplink power, the proposed scheme consistently achieves a higher overall CNE, demonstrating its superiority in closed-loop performance optimization.

We have also tested the other channel conditions, i.e., case III and case IV in Table \ref{tab2}. The proposed scheme consistently outperforms the baselines and exhibits similar trends to Figs. \ref{fig:simu_CNE_vs_B}–\ref{fig:simu_CNE_vs_pu}. Due to space limitations, the corresponding results are not shown here.

\section{Conclusion}
\label{Section:Conclusion}
In this paper, we have investigated the security of an NTN-assisted control system where a sensor, a UAV-mounted EIH, and a robot form an $\textbf{SC}^3$ closed loop to perform a critical task. The security implications of a potential eavesdropper have been considered. Under the constraint on the closed-loop eavesdropped information, we have formulated a CNE maximization problem that optimizes the bandwidth, power, and transmission times of the uplink and downlink transmissions, as well as the computing capability. For the superior legitimate channel case, we have proposed an efficient algorithm to find the globally optimal solution to the optimization problem based on the KKT conditions. For other cases, we have utilized MO theory to find the optimal solution. Simulation results have been provided to show the performance gain of the closed-loop structure-level PLS system design. Future work will consider robust designs under eavesdropper location uncertainty~\cite{position3}, extensions to multiple eavesdroppers, and the use of reconfigurable antenna architectures (e.g., rotatable antennas)~\cite{RA1,RA2} or integrated sensing and communications~\cite{ISAC} as additional spatial degrees of freedom for enhancing secrecy.

\appendices
\section{Proof of Proposition \ref{Prop1}}\label{Appendix_Prop1}
Based on \eqref{lemma1_2}, we can obtain 
\begin{equation}\label{A1_1}
	t^*_{\text{d}} = \frac{\rho R_{\text{u}}\left( p_{\text{u}},B_{\text{u}}\right)}{R_{\text{d}}\left( p_{\text{d}},B_{\text{d}}\right)}t^*_{\text{u}}.
\end{equation}	
Substituting \eqref{A1_1} into \eqref{security_equal}, we can obtain the optimal uplink time as
\begin{align}\label{A1_2}
	t^*_{\text{u}} = \frac{\frac{1}{\rho R_{\text{u}}\left(p_{\text{u}},B_{\text{u}} \right)} D_{\text{th}}}{\frac{R_{\text{SE}}\left(p_{\text{u}},B_{\text{u}} \right) }{R_{\text{u}}\left(p_{\text{u}},B_{\text{u}} \right)} +\frac{  R_{\text{CE}}\left(p_{\text{d}},B_{\text{d}} \right)}{R_{\text{\text{d}}}\left(p_{\text{d}},B_{\text{d}} \right)} }.
\end{align}
By substituting \eqref{A1_2} into \eqref{A1_1}, we can obtain \eqref{prop1_2} immediately. Based on \eqref{lemma1_1}, \eqref{prop1_1}, and \eqref{prop1_2}, the optimization problem \eqref{P1} can be reformulated as
\begin{subequations}\label{P2'}
	\begin{align} \max\limits_{\substack{p_{\text{u}},B_{\text{u}},p_{\text{d}},B_{\text{d}}}} &\frac{ D_{\text{th}}}{\frac{R_{\text{SE}}\left(p_{\text{u}},B_{\text{u}} \right) }{R_{\text{u}}\left(p_{\text{u}},B_{\text{u}} \right)} +\frac{  R_{\text{CE}}\left(p_{\text{d}},B_{\text{d}} \right)}{R_{\text{\text{d}}}\left(p_{\text{d}},B_{\text{d}} \right)} } \\
		\begin{split}\text{s.t.} \ &\left( \frac{1}{\rho R_{\text{u}}\left(p_{\text{u}},B_{\text{u}} \right)}+\frac{1}{R_{\text{d}}\left(p_{\text{d}},B_{\text{d}} \right)}+\frac{\alpha}{ \rho f_{\max}}\right) \\ &\quad\quad\quad\quad\quad\quad \times \frac{ D_{\text{th}}}{\frac{R_{\text{SE}}\left(p_{\text{u}},B_{\text{u}} \right) }{R_{\text{u}}\left(p_{\text{u}},B_{\text{u}} \right)} +\frac{  R_{\text{CE}}\left(p_{\text{d}},B_{\text{d}} \right)}{R_{\text{\text{d}}}\left(p_{\text{d}},B_{\text{d}} \right)} }  \leq T, \label{b6}\end{split}\\
		&0 \leq B_{\text{u}}\leq B_{\max},0 \leq B_{\text{d}}\leq B_{\max},  \label{f14}\\
		&0 \leq p_{\text{u}}\leq P_{\text{umax}}, 0 \leq p_{\text{d}}\leq P_{\text{dmax}}.
	\end{align}
\end{subequations}
As the numerator of the objective function in \eqref{P2'} is a constant parameter, maximizing the objective function is equivalent to minimizing the denominator term, indicating that the original problem \eqref{P1} is equivalent to \eqref{P2}.

\section{Proof of Proposition \ref{Prop2}}\label{Appendix_Prop2}
For notational convenience, we denote the objective function of \eqref{P2} as $f_{\text{obj}}\left( p_{\text{u}},B_{\text{u}},p_{\text{d}},B_{\text{d}}\right)$, and define an auxiliary function $g\left( p_{\text{u}},B_{\text{u}},p_{\text{d}},B_{\text{d}}\right)$ as 
\begin{align}\label{appendix_Prop2_1}
	&g\left( p_{\text{u}},B_{\text{u}},p_{\text{d}},B_{\text{d}}\right)\\\triangleq& \frac{D_{\text{th}}}{T}\left( \frac{1}{\rho R_{\text{u}}\left(p_{\text{u}},B_{\text{u}} \right)}+\frac{1}{R_{\text{d}}\left(p_{\text{d}},B_{\text{d}} \right)}+\frac{\alpha}{ \rho f_{\max}}\right),
\end{align}
It can be proven that $g\left( p_{\text{u}},B_{\text{u}},p_{\text{d}},B_{\text{d}}\right)$ is monotonically decreasing with respect to $p_{\text{u}}$, $B_{\text{u}}$, $p_{\text{d}}$, and $B_{\text{d}}$. With the above definition, the constraint \eqref{P2b} can be recast as
\begin{equation}\label{appendix_Prop2_2}
	f_{\text{obj}}\left( p_{\text{u}},B_{\text{u}},p_{\text{d}},B_{\text{d}}\right) \geq g\left( p_{\text{u}},B_{\text{u}},p_{\text{d}},B_{\text{d}}\right).
\end{equation}

Denoting the optimal solution of \eqref{P2} as $\left( p^*_{\text{u}},B^*_{\text{u}},p^*_{\text{d}},B^*_{\text{d}}\right)$, we have
\begin{equation}\label{appendix_Prop2_3}
	f_{\text{obj}}\left( p^*_{\text{u}},B^*_{\text{u}},p^*_{\text{d}},B^*_{\text{d}}\right) \geq g\left( p^*_{\text{u}},B^*_{\text{u}},p^*_{\text{d}},B^*_{\text{d}}\right)
\end{equation}
based on \eqref{P2b}.

First, we prove by contradiction that if $g_{\text{u}}>g_{\text{SE}}$, then $B_{\text{u}}^* = B_{\text{max}}$ at the optimal solution. In such a case, $f_{\text{obj}}\left( p_{\text{u}},B_{\text{u}},p_{\text{d}},B_{\text{d}}\right)$ is monotonically decreasing with respect to $B_{\text{u}}$ Based on \textbf{Lemma \ref{Lemma2}}. If $B_{\text{u}}^* < B_{\text{max}}$, we consider a new operating point $\left( p^*_{\text{u}}, B_{\text{max}},p^*_{\text{d}},B^*_{\text{d}}\right)$. Due to the decreasing monotonicity of $f_{\text{obj}}$ with respect to $B_{\text{u}}$, we can obtain that
\begin{equation}\label{appendix_Prop2_4}
f_{\text{obj}}\left( p^*_{\text{u}},B^*_{\text{u}},p^*_{\text{d}},B^*_{\text{d}}\right) > f_{\text{obj}}\left( p^*_{\text{u}},B_{\text{max}},p^*_{\text{d}},B^*_{\text{d}}\right).
\end{equation}
If the new point $\left( p^*_{\text{u}}, B_{\text{max}},p^*_{\text{d}},B^*_{\text{d}}\right)$ is feasible, i.e., $f_{\text{obj}}\left( p^*_{\text{u}},B_{\text{max}},p^*_{\text{d}},B^*_{\text{d}}\right) \geq g\left( p^*_{\text{u}},B_{\text{max}},p^*_{\text{d}},B^*_{\text{d}}\right)$, the new working point is a feasible point to \eqref{P2}, then it yields a strictly smaller objective function value than the assumed optimum, which is a contradiction.

Otherwise, if $f_{\text{obj}}\left( p^*_{\text{u}},B_{\text{max}},p^*_{\text{d}},B^*_{\text{d}}\right) < g\left( p^*_{\text{u}},B_{\text{max}},p^*_{\text{d}},B^*_{\text{d}}\right)$, which implies the new point is initially infeasible, we can improve the values of $p_{\text{u}}$, $p_{\text{d}}$, and $B_{\text{d}}$ to construct a new feasible point. Noting that we assume $\rho D_{\text{SE},0} + D_{\text{CE},0}> D_{\text{th}}$ (as discussed at the beginning of Section III), where $D_{\text{SE},0}$ and $D_{\text{CE},0}$ denote the corresponding eavesdropped data rate on the sensor-to-EIH link and the EIH-to-robot link when the security constraint is removed and resources are fully utilized (i.e., $p_{\text{u},0}=P_{\text{umax}}$, $p_{\text{d},0}=P_{\text{dmax}}$, and $B_{\text{u},0}=B_{\text{d},0}=B_{\text{max}}$). In this unconstrained scenario, the uplink and downlink transmission times can be denoted as~\cite{single_loop}
\begin{align}\label{appendix_Prop2_5}
	t_{\text{u},0} = \frac{\frac{1}{\rho R_{\text{u}}\left( P_{\text{umax}},B_{\text{max}}\right)}T}{\frac{1}{\rho R_{\text{u}}\left( P_{\text{umax}},B_{\text{max}}\right)}+\frac{1}{ R_{\text{d}}\left( P_{\text{dmax}},B_{\text{max}}\right)}+\frac{\alpha}{\rho f_{\text{max}}}},\\
	t_{\text{d},0} = \frac{\frac{1}{R_{\text{d}}\left( P_{\text{dmax}},B_{\text{max}}\right)}T}{\frac{1}{\rho R_{\text{u}}\left( P_{\text{umax}},B_{\text{max}}\right)}+\frac{1}{ R_{\text{d}}\left( P_{\text{dmax}},B_{\text{max}}\right)}+\frac{\alpha}{\rho f_{\text{max}}}}.
\end{align}	
Substituting the above values into the assumption $\rho D_{\text{SE},0} + D_{\text{CE},0}> D_{\text{th}}$ yields
\begin{align}\label{appendix_Prop2_6}
\frac{\frac{R_{\text{SE}}\left( P_{\text{umax}},B_{\text{max}}\right)}{R_{\text{u}}\left( P_{\text{umax}},B_{\text{max}}\right)}T+\frac{R_{\text{CE}}\left( P_{\text{dmax}},B_{\text{max}}\right)}{R_{\text{d}}\left( P_{\text{dmax}},B_{\text{max}}\right)}T}{\frac{1}{\rho R_{\text{u}}\left( P_{\text{umax}},B_{\text{max}}\right)}+\frac{1}{ R_{\text{d}}\left( P_{\text{dmax}},B_{\text{max}}\right)}+\frac{\alpha}{\rho f_{\text{max}}}}> D_{\text{th}},
\end{align}	
which can be recast using our defined functions as
\begin{equation}\label{appendix_Prop2_7}
	f_{\text{obj}}\left( P_{\text{umax}},B_{\text{max}},P_{\text{dmax}},B_{\text{max}}\right) > g\left( P_{\text{umax}},B_{\text{max}},P_{\text{dmax}},B_{\text{max}}\right).	
\end{equation}

Given $f_{\text{obj}}\left( p^*_{\text{u}},B_{\text{max}},p^*_{\text{d}},B^*_{\text{d}}\right) < g\left( p^*_{\text{u}},B_{\text{max}},p^*_{\text{d}},B^*_{\text{d}}\right)$ and \eqref{appendix_Prop2_7}, there must exist an intermediate feasible point $\left( p_{\text{u},1},B_{\text{max}},p_{\text{d},1},B_{\text{d},1}\right)$ with $p^*_{\text{u}} \leq p_{\text{u},1} \leq P_{\text{umax}}$, $p^*_{\text{d}} \leq p_{\text{d},1} \leq P_{\text{dmax}}$, and $B^*_{\text{d}} \leq B_{\text{d},1} \leq B_{\text{max}}$, such that
\begin{equation}\label{appendix_Prop2_8}
	f_{\text{obj}}\left( p_{\text{u},1},B_{\text{max}},p_{\text{d},1},B_{\text{d},1}\right) = g\left( p_{\text{u},1},B_{\text{max}},p_{\text{d},1},B_{\text{d},1}\right),
\end{equation}
due to the continuity of $f_{\text{obj}}$ and $g$.

For this feasible point $\left( p_{\text{u},1},B_{\text{max}},p_{\text{d},1},B_{\text{d},1}\right)$, and we have
\begin{align}
	f_{\text{obj}}\left( p_{\text{u},1},B_{\text{max}},p_{\text{d},1},B_{\text{d},1}\right) &= g\left( p_{\text{u},1},B_{\text{max}},p_{\text{d},1},B_{\text{d},1}\right)\label{appendix_Prop2_9_1}\\
	&< g\left( p^*_{\text{u}},B^*_{\text{u}},p^*_{\text{d}},B^*_{\text{d}}\right)\label{appendix_Prop2_9_2}\\
	&\leq f_{\text{obj}}\left( p^*_{\text{u}},B^*_{\text{u}},p^*_{\text{d}},B^*_{\text{d}}\right),\label{appendix_Prop2_9_3}
\end{align}
where \eqref{appendix_Prop2_9_2} holds because $g\left( p_{\text{u}},B_{\text{u}},p_{\text{d}},B_{\text{d}}\right)$ is monotonically decreasing with respect to $p_{\text{u}}$, $B_{\text{u}}$, $p_{\text{d}}$, and $B_{\text{d}}$, and \eqref{appendix_Prop2_9_3} follows from \eqref{appendix_Prop2_2}. This chain $f_{\text{obj}}\left( p_{\text{u},1},B_{\text{max}},p_{\text{d},1},B_{\text{d},1}\right)<f_{\text{obj}}\left( p^*_{\text{u}},B^*_{\text{u}},p^*_{\text{d}},B^*_{\text{d}}\right)$ conflicts the assumption that $\left( p^*_{\text{u}},B^*_{\text{u}},p^*_{\text{d}},B^*_{\text{d}}\right)$ is optimal. Therefore, when $g_{\text{u}}>g_{\text{SE}}$, the optimal uplink bandwidth must be $B_{\text{u}}^* = B_{\text{max}}$. The proofs for other cases ($g_{\text{u}}<g_{\text{SE}}$, $g_{\text{d}}>g_{\text{CE}}$, and $g_{\text{d}}<g_{\text{CE}}$) follow analogous arguments.

Next, we show that the \eqref{P2b} must hold with equality. Considering, for example, the case $g_{\text{u}}>g_{\text{SE}}$, we have
\begin{align}
	\lim\limits_{p_{\text{u}} \rightarrow 0}\frac{R_{\text{SE}}\left(p_{\text{u}},B_{\text{u}} \right) }{R_{\text{u}}\left(p_{\text{u}},B_{\text{u}} \right)} = \frac{g_{\text{SE}} }{g_{\text{u}}},\\
	\lim\limits_{p_{\text{u}} \rightarrow 0}\frac{1 }{\rho R_{\text{u}}\left(p_{\text{u}},B_{\text{u}} \right)} = +\infty.
\end{align}
Based on the above limits, when $p_{\text{u}} \rightarrow 0$, we can prove that the constraint \eqref{P2b} is not satisfied, as $f_{\text{obj}}$ is finite while $g +\infty$. The above analysis indicates that $p_{\text{u}}$ cannot be arbitrarily small due to the restriction of \eqref{P2b}. Now suppose at the optimum, we have $f_{\text{obj}}\left( p_{\text{u}},B_{\text{u}},p_{\text{d}},B_{\text{d}}\right) > g\left( p_{\text{u}},B_{\text{u}},p_{\text{d}},B_{\text{d}}\right)$ (i.e., \eqref{P2b} holds with strict inequality). We can decrease the value of $p_{\text{u}}$ so that the equality of \eqref{P2b} is satisfied. As the objective function is monotonically increasing with respect to $p_{\text{u}}$, the decrease of $p_{\text{u}}$ can always obtain a better solution. Similar reasoning applies by adjusting $p_{\text{d}}$, $B_{\text{u}}$, or $B_{\text{d}}$ for other case.

\section{Proof of Lemma \ref{Lemma3}}\label{Appendix_Lemma3}
The derivatives of $f_1 \left( p_{\text{u}} \right)$ and $h_1 \left( p_{\text{u}} \right)$ are
\begin{align}\label{A_Lemma3_1}
	f'_1 \left( p_{\text{u}} \right) &= \frac{\frac{a}{1+a p_{\text{u}}}\log \left( 1+b p_{\text{u}}\right) - \frac{b}{1+b p_{\text{u}}}\log \left( 1+a p_{\text{u}}\right)}{\log^2 \left( 1+b p_{\text{u}}\right)},\\
	h'_1 \left( p_{\text{u}} \right) &= -\frac{\log \left( 2 \right) D_{\text{th}}}{\rho TB_{\text{max}}}\frac{ \frac{b}{1+b p_{\text{u}}}}{\log^2 \left( 1+b p_{\text{u}}\right)},
\end{align}
where we denote $a = \frac{ g_{\text{SE}}}{B_{\text{max}N_0}}$ and $b = \frac{ g_{\text{u}}}{B_{\text{max}N_0}}$ for convenience. Therefore, the considered function can be calculated as
\begin{align}\label{A_Lemma3_2}
	\frac{f'_1 \left( p_{\text{u}} \right)}{h'_1 \left( p_{\text{u}} \right)} &=-\frac{\rho TB_{\text{max}}}{\log \left( 2 \right)\! D_{\text{th}}} \frac{\frac{a}{1+a p_{\text{u}}}\log \left( 1\! +\! a p_{\text{u}}\right)\! - \! \frac{b}{1+b p_{\text{u}}}\log \left( 1\! +\! b p_{\text{u}}\right)}{ \frac{b}{1+b p_{\text{u}}}},
\end{align}
whose derivative with respect to $p_{\text{u}}$ is
\begin{align}\label{A_Lemma3_3}
	\left[ \frac{f'_1 \left( p_{\text{u}} \right)}{h'_1 \left( p_{\text{u}} \right)}\right] ' &=\frac{\rho TB_{\text{max}}}{\log \left( 2 \right)\! D_{\text{th}}} \frac{a(a - b)\log(1+b p_{\text{u}})}{b(1+a p_{\text{u}})^2}.
\end{align}

When $g_{\text{u}}>g_{\text{SE}}$, we have $a<b$. Based on \eqref{A_Lemma3_3}, we have $	\left[ \frac{f'_1 \left( p_{\text{u}} \right)}{h'_1 \left( p_{\text{u}} \right)}\right] '<0$, indicating that $\frac{f'_1 \left( p_{\text{u}} \right)}{h'_1 \left( p_{\text{u}} \right)}$ is decreasing with respect to $p_{\text{u}}$. Following a similar procedure, we can prove that  $\frac{f'_2 \left( p_{\text{d}} \right)}{h'_2 \left( p_{\text{d}} \right)}$ is decreasing with respect to $p_{\text{d}}$, which completes the proof.

\section{Proof of Proposition \ref{Prop3}}\label{Appendix_Prop3}
According to \textbf{Lemma \ref{Lemma3}}, \eqref{equation1} defines a relationship as
\begin{align}\label{A_Prop3_1}
	p_{\text{d}} = C_1 \left( p_{\text{u}} \right),
\end{align}
where $C_1 \left( p_{\text{u}} \right)$ is an increasing function of $p_{\text{u}}$.

In addition, as the function $	f_1 \left( p_{\text{u}} \right)-h_1 \left( p_{\text{u}} \right)$ and $f_2 \left( p_{\text{d}} \right)-h_2 \left( p_{\text{d}} \right)$ are increasing with respect to $p_{\text{u}}$ and $p_{\text{d}}$, respectively, expression \eqref{equation2} can be expressed as
\begin{align}\label{A_Prop3_2}
	p_{\text{d}} = C_2 \left( p_{\text{u}} \right),
\end{align}
where $C_2\left( p_{\text{u}} \right)$ is a decreasing function of $p_{\text{u}}$. Since $C_1 \left( p_{\text{u}} \right)$ is increasing and $C_2 \left( p_{\text{u}} \right)$ is decreasing, system of equations, \eqref{equation1} and \eqref{equation2}, can have at most one intersection point, denoted as $\left( \hat{p_{\text{u}}}, \hat{p_{\text{d}}} \right)$ if it exists.

Next, we analyze the objective function of problem \eqref{P3}, $f_1\left( p_{\text{u}} \right) + f_2\left(  p_{\text{d}} \right)$, along the curve defined by equation \eqref{equation2}, i.e., $p_{\text{d}} = C_2 \left( p_{\text{u}} \right)$. The constrained objective function value, denoted as $O\left(p_{\text{u}}\right)  $, can be written as
\begin{align}\label{A_Prop3_3}
O\left(p_{\text{u}}\right) = f_1\left( p_{\text{u}} \right) + f_2\left( C_2 \left( p_{\text{u}} \right) \right),
\end{align}
whose derivative with respect to $p_{\text{u}}$ can be calculated as
\begin{align}\label{A_Prop3_4}
	\frac{\mathrm{d}O\left(p_{\text{u}}\right)}{\mathrm{d}p_{\text{u}}} = f'_1\left( p_{\text{u}} \right) + f'_2\left( C_2 \left( p_{\text{u}} \right) \right)\frac{\mathrm{d}C_2\left(p_{\text{u}}\right)}{\mathrm{d}\left(p_{\text{u}}\right)}.
\end{align}
Utilizing implicit differentiation of \eqref{equation2}, the derivative of $C_2 \left( p_{\text{u}} \right)$ can be calculated as
\begin{align}\label{A_Prop3_5}
	\frac{\mathrm{d}C_2\left(p_{\text{u}}\right)}{\mathrm{d}p_{\text{u}}} =-\frac{f'_1\left( p_{\text{u}} \right)-h'_1\left( p_{\text{u}} \right)}{f'_2\left( C_2 \left( p_{\text{u}} \right) \right)-h'_2\left(C_2 \left( p_{\text{u}} \right) \right)}.
\end{align}

Substituting \eqref{A_Prop3_5} into \eqref{A_Prop3_4}, we have
\begin{align}
	\frac{\mathrm{d}O\left(p_{\text{u}}\right)}{\mathrm{d}p_{\text{u}}} &= \frac{f'_2\left( C_2 \left( p_{\text{u}} \right) \right)h'_1\left( p_{\text{u}} \right)-f'_1\left( p_{\text{u}} \right)h'_2\left( C_2 \left( p_{\text{u}} \right) \right)}{f'_2\left( C_2 \left( p_{\text{u}} \right) \right)-h'_2\left(C_2 \left( p_{\text{u}} \right) \right)}\label{A_Prop3_6_a}	\\
	& = \frac{f'_2\left( p_{\text{d}}  \right)h'_1\left( p_{\text{u}} \right)-f'_1\left( p_{\text{u}} \right)h'_2\left(  p_{\text{d}} \right)}{f'_2\left(  p_{\text{d}}  \right)-h'_2\left(p_{\text{d}}\right)}\label{A_Prop3_6_b}	\\
	& =\frac{\frac{f'_2 \left( p_{\text{d}} \right)}{h'_2 \left( p_{\text{d}} \right)}-	\frac{f'_1 \left( p_{\text{u}} \right)}{h'_1 \left( p_{\text{u}} \right)}}{f'_2\left(  p_{\text{d}}  \right)-h'_2\left(p_{\text{d}}\right)}h'_1\left( p_{\text{u}} \right)h'_2\left( p_{\text{d}} \right),\label{A_Prop3_6_c}	
\end{align}
where  $p_{\text{d}} = C_2 \left( p_{\text{u}} \right)$ is used in \eqref{A_Prop3_6_b} for brevity.
Based on the fact that $f'_1 \left( p_{\text{u}} \right)>0$, $f'_2 \left( p_{\text{d}} \right)>0$,$h'_1 \left( p_{\text{u}} \right)<0$, and $h'_2 \left( p_{\text{d}} \right)<0$, the sign of $\frac{\mathrm{d}O\left(p_{\text{u}}\right)}{\mathrm{d}p_{\text{u}}}$ is determined by the sign of $\frac{f'_2 \left( p_{\text{d}} \right)}{h'_2 \left( p_{\text{d}} \right)}-	\frac{f'_1 \left( p_{\text{u}} \right)}{h'_1 \left( p_{\text{u}} \right)}$. Specifically, when $\frac{f'_1 \left( p_{\text{u}} \right)}{h'_1 \left( p_{\text{u}} \right)}>\frac{f'_2 \left( p_{\text{d}} \right)}{h'_2 \left( p_{\text{d}} \right)}$, we have $\frac{\mathrm{d}O\left(p_{\text{u}}\right)}{\mathrm{d}p_{\text{u}}}<0$, and when  $\frac{f'_1 \left( p_{\text{u}} \right)}{h'_1 \left( p_{\text{u}} \right)}<\frac{f'_2 \left( p_{\text{d}} \right)}{h'_2 \left( p_{\text{d}} \right)}$, we have $\frac{\mathrm{d}O\left(p_{\text{u}}\right)}{\mathrm{d}p_{\text{u}}}>0$.

At the intersection point $\left( \hat{p_{\text{u}}}, \hat{p_{\text{d}}} \right)$, we have $\frac{f'_1 \left( \hat{p_{\text{u}}} \right)}{h'_1 \left(\hat{p_{\text{u}}} \right)}=\frac{f'_2 \left( \hat{p_{\text{d}}} \right)}{h'_2 \left( \hat{p_{\text{d}}} \right)}$. Next, we consider a point $\left( {p_{\text{u}}}, {p_{\text{d}}} \right)$ on curve $p_{\text{d}} = C_2 \left( p_{\text{u}} \right)$. If $p_{\text{u}}<\hat{p_{\text{u}}}$, we have $p_{\text{d}}=C_2 \left( p_{\text{u}} \right)>C_2 \left( \hat{p_{\text{u}}}\right)=\hat{p_{\text{d}}}$ as $C_2$ is decreasing. Based on \textbf{Lemma \ref{Lemma3}}, $\frac{f'_1 \left( p_{\text{u}} \right)}{h'_1 \left( p_{\text{u}} \right)}$ and $\frac{f'_2 \left( p_{\text{d}} \right)}{h'_2 \left( p_{\text{d}} \right)}$ are decreasing functions. Therefore, we have $\frac{f'_1 \left( p_{\text{u}} \right)}{h'_1 \left( p_{\text{u}} \right)}>\frac{f'_1 \left( \hat{p_{\text{u}}} \right)}{h'_1 \left( \hat{p_{\text{u}}} \right)} = \frac{f'_2 \left( \hat{p_{\text{d}}} \right)}{h'_2 \left( \hat{p_{\text{d}}} \right)}>\frac{f'_2 \left( p_{\text{d}} \right)}{h'_2 \left( p_{\text{d}} \right)}$, which indicates that $\frac{\mathrm{d}O\left(p_{\text{u}}\right)}{\mathrm{d}p_{\text{u}}}<0$. Thus, $O\left(p_{\text{u}}\right) $ is decreasing for $p_{\text{u}}<\hat{p_{\text{u}}}$. Similarly, when $p_{\text{u}}>\hat{p_{\text{u}}}$, it can shown that $O\left(p_{\text{u}}\right)$ is increasing with respect to $p_{\text{u}}$.  

Based on the above analysis, $O\left(p_{\text{u}}\right)$ has a unique minimum at  $p_{\text{u}}= \hat{p_{\text{u}}}$ along the curve $p_{\text{d}} = C_2 \left( p_{\text{u}} \right)$. Therefore, $\left( \hat{p_{\text{u}}}, \hat{p_{\text{d}}} \right)$ is the point to minimize the objective function $f_1\left( p_{\text{u}} \right) + f_2\left( p_{\text{d}} \right) $ subject to constraint \eqref{equation2} (i.e., \eqref{P3b} holding with equality). According to \textbf{Proposition \ref{Prop2}}, we know the optimal solution to \eqref{P3} must satisfy constraint \eqref{P3b} with equality. Therefore, if the point $\left( \hat{p_{\text{u}}}, \hat{p_{\text{d}}} \right)$ is feasible with $\hat{p_{\text{u}}}\leq P_{\text{umax}}$ and $\hat{p_{\text{d}}}\leq P_{\text{dmax}}$, it is the global solution to problem \eqref{P3}. 

Otherwise, if the intersection does not exist, or if $\left( \hat{p_{\text{u}}}, \hat{p_{\text{d}}} \right)$ is not feasible, then $O\left(p_{\text{u}}\right)$ is monotonous within the interval $p_{\text{u}} \in \left[ C_2^{-1}\left( P_{\text{dmax}}\right) ,P_{\text{umax}}\right] $, where $C_2^{-1}$ denotes the inverse function of $C_2$. In such a case, the minimum of $O\left(p_{\text{u}}\right)$ must occur at one of the boundary points corresponding to either $p_{\text{u}}=P_{\text{umax}}$ or $p_{\text{d}}=P_{\text{dmax}}$. The boundary point yielding the smaller objective function value is the optimal solution.


\newpage

%If you have an EPS/PDF photo (graphicx package needed), extra braces are needed around the contents of the optional argument to biography to prevent the LaTeX parser from getting confused when it sees the complicated
%$\backslash${\tt{includegraphics}} command within an optional argument. (You can create your own custom macro containing the $\backslash${\tt{includegraphics}} command to make things simpler here.)
 

%\bf{If you include a photo:}\vspace{-33pt}
%\begin{IEEEbiography}[{\includegraphics[width=1in,height=1.25in,clip,keepaspectratio]{fig1}}]{Michael Shell}
%Use $\backslash${\tt{begin\{IEEEbiography\}}} and then for the 1st argument use $\backslash${\tt{includegraphics}} to declare and link the author photo.
%Use the author name as the 3rd argument followed by the biography text.
%\end{IEEEbiography}

\vspace{11pt}

\begin{IEEEbiography}[{\includegraphics[width=1in,height=1.25in,clip,keepaspectratio]{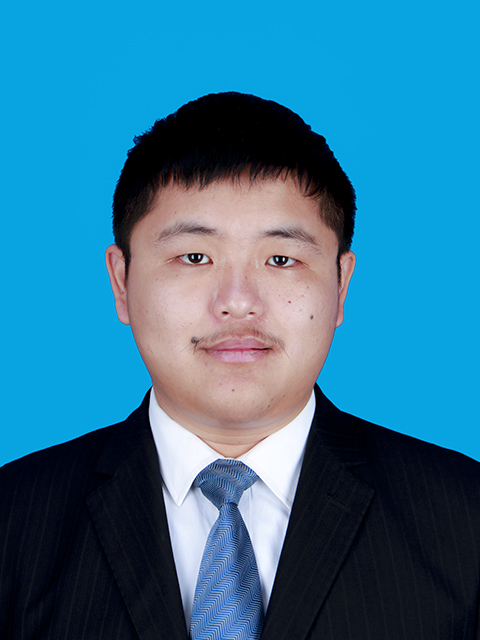}}]{Chengleyang Lei}
	received the B.S. degree in 2021 from the Department of Electronic Engineering, Tsinghua University, Beijing, China. He is currently pursuing the Ph.D degree with the Department of Electronic Engineering in Tsinghua University. His research interests include the communication and control integration, coordinated satellite-UAV-terrestrial networks, and future 6G technologies.
\end{IEEEbiography}

\begin{IEEEbiography}[{\includegraphics[width=1in,height=1.25in,clip,keepaspectratio]{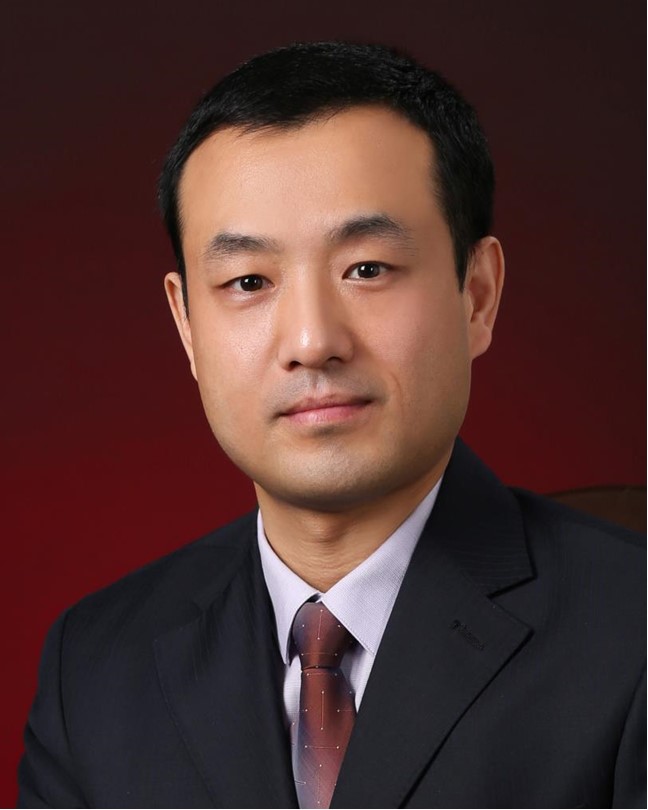}}]{Wei Feng}(Senior Member, IEEE) received the B.S. and Ph.D. degrees from the Department of Electronic Engineering, Tsinghua University, Beijing, China, in 2005 and 2010, respectively. He is currently a professor with the Department of Electronic Engineering, Tsinghua University. He also serves as Vice Dean of the Shuimu College, Tsinghua University, and Chief Scientist of Network Science with the State Key Laboratory of Space Network and Communications, Beijing, China. His research interests include space-air-ground integrated networks, 6G mobile communications, maritime Internet of things, and Internet of intelligent robots. Dr. Feng has received the National Technological Invention Award of China in 2016, the Outstanding Young Scholars Fund of Natural Science Foundation of China (NSFC) in 2019, and the Distinguished Young Scholars Fund of NSFC in 2024. He currently serves as the Assistant to the Editor-in-Chief of \textsc{China Communications}, and an Associate Editor for \textsc{IEEE Transactions on Aerospace and Electronic Systems}. He served as an Editor for \textsc{IEEE Transactions on Cognitive Communications and Networking} from 2019 to 2023. He is a Fellow of the China Institute of Communications.
\end{IEEEbiography}

\begin{IEEEbiography}[{\includegraphics[width=1in,height=1.25in,clip,keepaspectratio]{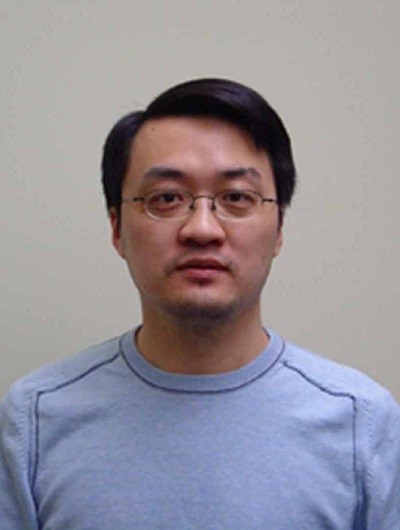}}]{Yunfei Chen}(Fellow, IEEE)
	received the B.E. and M.E. degrees in electronics engineering from Shanghai Jiaotong University, Shanghai, China, in 1998 and 2001, respectively, and the Ph.D. degree from the University of Alberta in 2006. 
	He is currently working as a Professor with the Department of Engineering, University of Durham, U.K. 
	His research interests include wireless communications, cognitive radios, wireless relaying, and energy harvesting.
\end{IEEEbiography}

\begin{IEEEbiography}[{\includegraphics[width=1in,height=1.25in,clip,keepaspectratio]{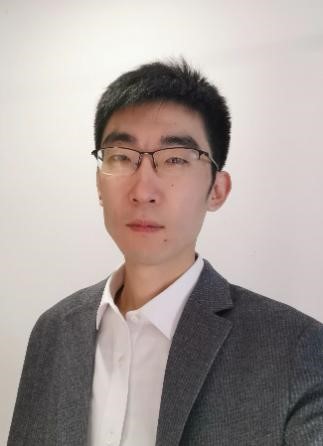}}]{Jue Wang}(Member, IEEE)
received the B.S. degree in communications engineering from Nanjing University, Nanjing, China, in 2006, and the M.S. and Ph.D. degrees from the National Communications Research Laboratory, Southeast University, Nanjing, in 2009 and 2014, respectively. From 2014 to 2016, he was with Singapore University of Technology and Design as a Post-Doctoral Research Fellow. He is currently with the School of Information Science and Technology, Nantong University, Nantong, China. His research interests include MIMO, RIS, NTN, and machine learning in communications. He has served as a technical program committee member and a reviewer for a number of IEEE conferences/journals. He was awarded as an Exemplary Reviewer of \textsc{IEEE Transactions on Communications} in 2014 and an Exemplary Reviewer of \textsc{IEEE Wireless Communications Letters} in 2021.
\end{IEEEbiography}

\begin{IEEEbiography}[{\includegraphics[width=1in,height=1.25in,clip,keepaspectratio]{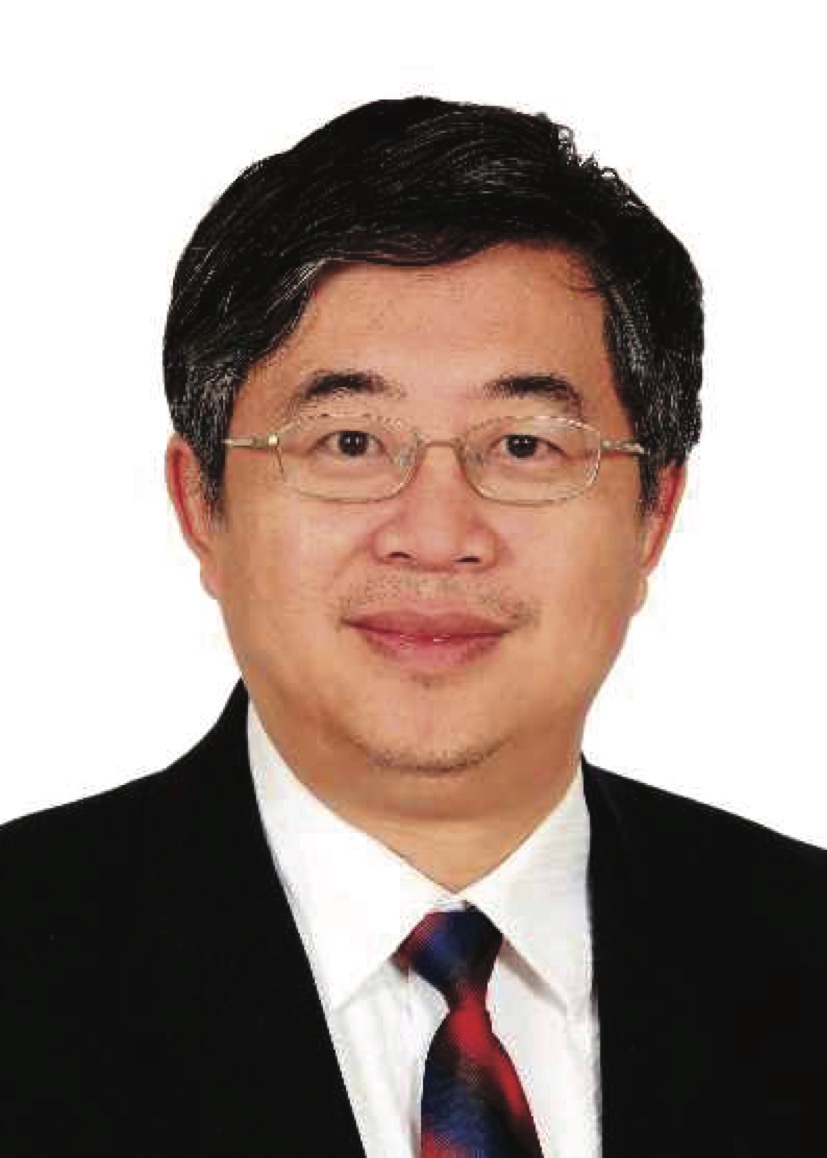}}]{Ning Ge}(Member, IEEE)
	received the B.S. and Ph.D. degrees from Tsinghua University, Beijing, China, in 1993 and 1997, respectively.
	From 1998 to 2000, he was with ADC Telecommunications, Dallas, TX, USA, where he researched the development of ATM switch fabric ASIC. 
	Since 2000, he has been a Professor with the Department of Electronics Engineering, Tsinghua University. 
	He has published over 100 papers. 
	His current research interests include communication ASIC design, short range wireless communication, and wireless communications.
	Dr. Ge is a senior member of the China Institute of Communications and the Chinese Institute of Electronics.
\end{IEEEbiography}

\begin{IEEEbiography}[{\includegraphics[width=1in,height=1.25in,clip,keepaspectratio]{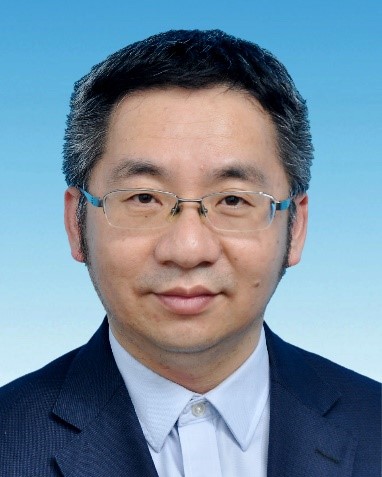}}]{Shi Jin}(Fellow, IEEE)
	received the B.S. degree in communications engineering from Guilin University of Electronic Technology, Guilin, China, in 1996, the M.S. degree from Nanjing University of Posts and Telecommunications, Nanjing, China, in 2003, and the Ph.D. degree in information and communications engineering from the Southeast University, Nanjing, in 2007. From June 2007 to October 2009, he was a Research Fellow with the Adastral Park Research Campus, University College London, London, U.K. He is currently with the Faculty of the School of Information Science and Engineering, Southeast University. His research interests include wireless communications, random matrix theory, and information theory. 
	
	Dr. Jin is serving as an Area Editor for the \textsc{IEEE Transactions on Communications} and \textit{IET Electronics Letters}. He was an Associate Editor for the \textsc{IEEE Transactions on Wireless Communications}, and \textsc{IEEE Communications Letters}, and \textit{IET Communications}. Dr. Jin and his coauthors have been awarded the IEEE Communications Society Stephen O. Rice Prize Paper Award in 2011, the IEEE Jack Neubauer Memorial Award in 2023, The IEEE Marconi Prize Paper Award in Wireless Communications in 2024, and the IEEE Signal Processing Society Young Author Best Paper Award in 2010 and Best Paper Award in 2022. 
\end{IEEEbiography}

\begin{IEEEbiography}[{\includegraphics[width=1in,height=1.25in,keepaspectratio]{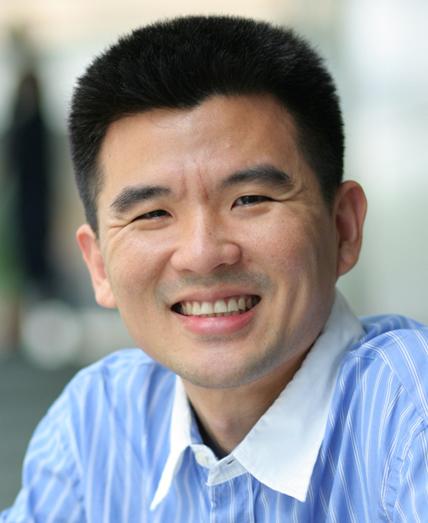}}]
	{Tony Q.S. Quek}(Fellow, IEEE) received the B.E.\ and M.E.\ degrees in electrical and electronics engineering from the Tokyo Institute of Technology in 1998 and 2000, respectively, and the Ph.D.\ degree in electrical engineering and computer science from the Massachusetts Institute of Technology in 2008. Currently, he is the Associate Provost (AI \& Digital Innovation) and Cheng Tsang Man Chair Professor with Singapore University of Technology and Design (SUTD). He also serves as the Director of the Future Communications R\&D Programme, and the ST Engineering Distinguished Professor. He is a co-founder of Silence Laboratories and NeuroRAN. His current research topics include wireless communications and networking, network intelligence, non-terrestrial networks, open radio access network, AI-RAN, and 6G.
	
	Dr.\ Quek was honored with the 2008 Philip Yeo Prize for Outstanding Achievement in Research, the 2012 IEEE William R. Bennett Prize, the 2015 SUTD Outstanding Education Awards -- Excellence in Research, the 2016 IEEE Signal Processing Society Young Author Best Paper Award, the 2017 CTTC Early Achievement Award, the 2017 IEEE ComSoc AP Outstanding Paper Award, the 2020 IEEE Communications Society Young Author Best Paper Award, the 2020 IEEE Stephen O. Rice Prize, the 2020 Nokia Visiting Professor, the 2022 IEEE Signal Processing Society Best Paper Award, the 2024 IIT Bombay International Award For Excellence in Research in Engineering and Technology, the IEEE Communications Society WTC Recognition Award 2024, and the Public Administration Medal (Bronze). He is an IEEE Fellow, a WWRF Fellow, an AIIA Fellow, a member of NAAI, and a Fellow of the Academy of Engineering Singapore.
\end{IEEEbiography}

\vfill

\end{document}